\documentclass[final,1p,times]{elsarticle}

\usepackage{amsmath,amssymb,amsfonts,amscd}
\usepackage[T1]{fontenc}

\usepackage{amsmath,amsthm,amssymb}

\usepackage{graphicx}

\def\R{\mathbb R}
\def\N{\mathbb N}
\def\C{\mathbb C}

\renewcommand{\d}{{\mathrm d}}
\def\D{\mathrm d} 
\renewcommand{\i}{{\mathrm i}}
\def\e{\mathrm e}
\def\O{\mathcal O}
\def\OO{\mathcal O}
\def\G{\mathcal{G}}
\def\GG{\mathcal{G}} 
\def\HH{\mathcal{H}}
\def\rank{\mathrm{rank}}

\renewcommand{\Re}{\mathrm{Re\,}}

\newtheorem{lem}{Lemma}[section]

\newtheorem{theorem}[lem]{Theorem}
\newtheorem{coro}[lem]{Corollary}

\newtheorem{conv}[lem]{Convention}

\theoremstyle{definition}

\newtheorem{rem}[lem]{Remark}

\begin{document}

\begin{frontmatter}

\title{Approximation of a general singular vertex coupling\\ in quantum graphs}
\author[label1]{Taksu Cheon} \ead{taksu.cheon@kochi-tech.ac.jp}
\author[label2,label3]{Pavel Exner} \ead{exner@ujf.cas.cz}
\author[label2,label4]{Ond\v{r}ej Turek} \ead{turekond@fjfi.cvut.cz}
\address[label1]{Laboratory of Physics, Kochi University of Technology\\
Tosa Yamada, Kochi 782-8502, Japan}
\address[label2]{Doppler Institute for Mathematical Physics and Applied
Mathematics, Czech Technical University\\
B\v rehov{\'a} 7, 11519 Prague, Czech Republic}
\address[label3]{Department of Theoretical Physics, Nuclear Physics
Institute, Czech Academy of Sciences\\ 
25068 \v{R}e\v{z} near Prague, Czech Republic}
\address[label4]{Department of Mathematics, Faculty of Nuclear Sciences and
Physical Engineering, Czech Technical University\\
Trojanova 13, 12000 Prague, Czech Republic}

\date{\today}

\begin{abstract}
The longstanding open problem of approximating all
singular vertex couplings in a quantum graph is solved. We present
a construction in which the edges are decoupled; an each pair of
their endpoints is joined by an edge carrying a $\delta$ potential
and a vector potential coupled to the ``loose'' edges by a
$\delta$ coupling. It is shown that if the lengths of the
connecting edges shrink to zero and the potentials are properly
scaled, the limit can yield any prescribed singular vertex
coupling, and moreover, that such an approximation converges in
the norm-resolvent sense.
\end{abstract}


\end{frontmatter}

\section{Introduction}

While the origin of the idea to investigate quantum mechanics of
particles confined to a graph was conceived originally to address
to a particular physical problem, namely the spectra of aromatic
hydrocarbons \cite{RS53}, the motivation was quickly lost and for a
long time the problem remained rather an obscure textbook example.
This changed in the last two decades when the progress of
microfabrication techniques made graph-shaped structures of
submicron sizes technologically important. This generated an
intense interest to investigation of quantum graph models which went
beyond the needs of practical applications, since these models
proved to be an excellent laboratory to study various properties
of quantum systems. The literature on quantum graphs is nowadays
huge; we limit ourselves to mentioning the recent volume
\cite{EKST08} where many concepts are discussed and a rich
bibliography can be found.

The essential component of quantum graph models is the wavefunction
coupling in the vertices. While often the most simple matching
conditions (dubbed free, Kirchhoff, or Neumann) or the slightly more general $\delta$ coupling in which the functions are continuous in the vertex are used, these cases represent just a tiny subset
of all admissible couplings. The family of the latter is
determined by the requirement that the corresponding Hamiltonian
is a self-adjoint operator, or in physical language, that the
probability current is conserved at the vertices. It is not
difficult to find all the admissible conditions mathematically; if
the vertex joins $n$ edges they contain $n^2$ free parameters, and
with exception of the one-parameter subfamily mentioned above they
are all singular in the sense that the wavefunctions are
discontinuous at the vertex.

What is much less clear is the physical meaning of such
conditions. It is longstanding open problem whether and in what
sense one can approximate all the singular couplings by regular
ones depending on suitable parameters, and the aim of the present
paper is to answer this question by presenting such a
construction, minimal in a natural sense using $n^2$ real
parameters, and to show that the closeness is achieved in the
norm-resolvent sense, so the convergence of all types of the
spectra and the corresponding eigenprojections is guaranteed.

The key idea comes from a paper of one of us with Shigehara
\cite{CS98} which showed that a combination of regular point
interactions on a line approaching each other with the coupling
scaled in a particular way w.r.t. the interaction distance can
produce a singular point interaction. Later it was demonstrated
\cite{ENZ01} that the convergence in this model is norm-resolvent
and the scaling choice is highly non-generic. The idea was applied
by two of us to the simplest singular coupling, the so-called
$\delta'_\mathrm{s}$, in \cite{CE04} and was demonstrated to work; the
question was how much it can be extended. Two other of us examined
it \cite{ET07} and found that with a larger number of regular
interactions one can deal with families described by $2n$
parameters, and changing locally the approximating graph topology
one can deal with all the couplings invariant with respect to the time
reversal which form an ${n+1\choose 2}$-parameter subset.

It was clear that to proceed beyond the time-reversal symmetry one
has to involve vector potentials similarly as it is was done in
the simplest situation in \cite{SMMC99}. In this paper we present
such a construction which contains parameters breaking the
symmetry and which at the same time is more elegant than that of
\cite{ET07} in the sense that the needed ``ornamentation'' of the
graph is minimal: we disconnect the $n$ edges at the vertex and
join each pair of the so obtained free ends by an additional edge
which shrinks to a point in the limit. The number of parameters
leans on the decomposition $n^2 = n + 2 {n\choose 2}$, where the first summand, $n$,
corresponds to $\delta$ couplings of the ``outer'' edge endpoints
with those of the added shrinking ones. The second summand can be considered as ${n\choose 2}$ times
two parameters: one is a $\delta$ potential placed at the edge,
the other is a vector potential supported by it.

Our result shows that any singular vertex coupling can be
approximated by a graph in which the vertex is replaced by a local
graph structure in combination with local regular interactions and
local magnetic fields. This opens way to constructing
``structured'' vertices tailored to the desired conductivity
properties, even tunable ones, if the interactions are controlled
by gate electrodes, however, we are not going to elaborate such
approximations further in this paper.

We have to note for completeness that the problem of understanding
vertex couplings has also other aspects. The approximating object
needs not to be a graph but can be another geometrical structure. A
lot of attention was paid to the situation of ``fat graphs'', or
networks of this tubes built around the graph skeleton. The two
approaches can be combined, for instance, by ``lifting'' the graph
results to fat graphs. In this way approximations to $\delta$ and
$\delta'_\mathrm{s}$ couplings by suitable families of
Schr\"odinger operators on such manifolds with Neumann boundaries
were recently demonstrated in \cite{EP08}. The results of this
paper can be similarly ``lifted'' to manifolds; that will be the
subject of a subsequent work.

Let us review briefly the contents of the paper. In the next
section we gather the needed preliminary information. We review
the information about vertex couplings and derive a new
parametrization of a general coupling suitable for our purposes.
In Section~3 we describe in detail the approximation sketched
briefly above and show that on a heuristic level it converges to a
chosen vertex coupling. Finally, in the last section we present
and prove our main result showing that the said convergence is not
only formal but it is valid also in the norm-resolvent sense.

\section{Vertex coupling in quantum graphs}

Let us first recall briefly a few basic notions; for a more
detailed discussion we refer to the literature given in the
introduction. The object of our interest are Schr\"odinger
operators on metric graphs. A graph is conventionally identified
with a family of vertices and edges; it is metric if each edge can
be equipped with a distance, i.e. to be identified with a finite
or semi-infinite interval.

We regard such a graph $\Gamma$ with edges $E_1,\dots,E_n$ as a
configuration space of a quantum mechanical system, i.e. we
identify the orthogonal sum $\HH= \bigoplus_{j=1}^n L^2(E_j)$ with
the state Hilbert space and the wave function of a spinless
particle ``living'' on $\Gamma$ can be written as the column
$\Psi=(\psi_1, \psi_2,\ldots,\psi_n)^T$ with $\psi_j\in L^2(E_j)$.
In the simplest case when no external fields are present the
system Hamiltonian acts as $(H_\Gamma \Psi)_j= -\psi''_j$, with
the domain consisting of functions from $W^{2,2}(\Gamma) :=
\bigoplus_{j=1}^n W^{2,2}(E_j)$. Not all such functions are
admissible, though, in order to make the operator self-adjoint we
have to require that appropriate boundary conditions are satisfied
at the vertices of the graph.

We restrict our attention to the physically most interesting case
when the boundary conditions are \emph{local}, coupling values of
the functions and derivatives is each vertex separately. Our aim
is explain the meaning of a general vertex coupling using suitable
approximations; the local character means that we can investigate
how such a system behaves in the vicinity of a single vertex. A
prototypical example of this situation is a \emph{star graph} with
one vertex in which a finite number of semi-infinite edges meet;
this is the case we will mostly have in mind in the following.

Let us thus consider a graph vertex $V$ of degree $n$, i.e. with
$n$ edges connected at $V$. We denote these edges by
$E_1,\ldots,E_n$ and the components of the wave function values at
them by $\psi_1(x_1),\ldots,\psi_n(x_n)$. We choose the
coordinates at the edges in such a way that $x_j\geq0$ for all
$j=1,\ldots,n$, and the value $x_j=0$ corresponds to the vertex
$V$. For notational simplicity we put $\Psi_V=(\psi_1(0),
\ldots,\psi_n(0))^T$ and $\Psi'_V=(\psi'_1(0), \ldots,
\psi'_n(0))^T$. Since our Hamiltonian is a second-order
differential operator, the sought boundary conditions will couple
the above boundary values, their most general form being
\begin{equation}\label{1}
A\Psi_V+B\Psi'_V=0\,,
\end{equation}
where $A$ and $B$ are complex $n\times n$ matrices.

To ensure self-adjointness of the Hamiltonian, which is in physical
terms equivalent to conservation of the probability current at the
vertex $V$, the matrices $A$ and $B$ cannot be arbitrary but have
to satisfy the following two conditions,
\begin{equation}\label{KS}
\begin{split}
\bullet \quad & \mathrm{rank}(A|B)=n,\\
\bullet \quad & \text{the matrix $AB^*$ is self-adjoint},
\end{split}
\end{equation}
where $(A|B)$ denotes the $n\times2n$ matrix with $A,B$ forming
the first and the second $n$ columns, respectively, as stated for
the first time by Kostrykin and Schrader \cite{KS99}. The relation
\eqref{1} together with conditions \eqref{KS} (for brevity, we
will write \eqref{1}\&\eqref{KS}) describe all possible vertex
boundary conditions giving rise to a self-adjoint Hamiltonian; we
will speak about \emph{admissible boundary conditions}.

On the other hand, it is obvious that the formulation
\eqref{1}\&\eqref{KS} is non-unique in the sense that different
pairs $(A_1,B_1)$, $(A_2,B_2)$ may define the same vertex
coupling, as $A,B$ can be equivalently replaced by $CA,CB$ for any
regular matrix $C\in\C^{n,n}$. To overcome this ambiguity, Harmer
\cite{Ha00}, and independently Kostrykin and Schrader \cite{KS00}
proposed a unique form of the boundary conditions \eqref{1},
namely
\begin{equation}\label{U}
(U-I)\Psi_V+\i(U+I)\Psi'_V=0\,,
\end{equation}
where $U$ is a unitary $n\times n$ matrix. Note that in a more
general context such conditions were known before \cite{GG91}, see also \cite{FT00}.

The natural \emph{parametrization} \eqref{U} of the family of
vertex couplings has several advantages in comparison to
\eqref{1}\&\eqref{KS}, besides its uniqueness it also makes
obvious how ``large'' the family is: since the unitary group
$U(n)$ has $n^2$ real parameters, the same is true for vertex
couplings in a quantum graph vertex of the degree $n$. Of course,
this fact is also clear if one interprets the couplings from the
viewpoint of self-adjoint extensions \cite{ES89}.

On the other hand, among the disadvantages of the formulation
\eqref{U} one can mention its complexity: vertex couplings that
are simple from the physical point of view may have a complicated
description when expressed in terms of the condition \eqref{U}. As
an example, let us mention in the first place the
$\delta$-coupling with a parameter $\alpha\in\R$, characterized by
relations
\begin{equation}\label{delta}
\psi_j(0)=\psi_k(0)=:\psi(0)\,, \quad j,k=1\ldots,n\,, \qquad
\sum^{n}_{j=1}\psi_j'(0)=\alpha\psi(0)\,,
\end{equation}
for which the matrix $U$ used in \eqref{U} has entries given by
\begin{equation}\label{Udelta}
U_{jk}=\frac{2}{n+\i\alpha}-\delta_{jk}\,,
\end{equation}
$\delta_{jk}$ being the Kronecker delta. When we substitute
\eqref{Udelta} into \eqref{U} and compare with \eqref{delta}
rewritten into a matrix form~\eqref{1}, we observe that the first
formulation is not only more complicated with respect to the
latter, but also contains complex values whereas the latter does
not. This is a reason why it is often better to work with simpler
expressions of the type \eqref{1}\&\eqref{KS}. Another aspect of
this parametrization difference concerns the meaning of the
parameters. Since the $n^2$ ones mentioned earlier are
``encapsulated'' in a unitary matrix, it is difficult to
understand which role each of them plays.

On the other hand, both formulations \eqref{1}\&\eqref{KS} and
\eqref{U} have a common feature, namely that they have a form
insensitive to a particular edge numbering. If the edges are
permuted one has just to replace the matrices $A,B$ and $U$ by
$\tilde{A},\tilde{B}$ and $\tilde{U}$, respectively, obtained by
the appropriate rearrangement of rows and columns. This may hide
different ways in which the edges are coupled; it is easy to see
that a particular attention should be paid to ``singular''
situations when the matrix $U$ has eigenvalue(s) equal to $\pm 1$.

Since the type of the coupling will be important for the
approximation we are going to construct, we will rewrite the vertex
coupling conditions in another form which is again simple and
unique but requires an appropriate edge numbering. This will be
done in Theorem~\ref{volneOP}, before stating it we introduce several
symbols that will be employed in the further text, namely
\begin{align*}
\C^{k,l}\quad -&\quad \text{the set of complex matrices with $k$
rows and $l$ columns,} \\
\hat{n}\quad -&\quad \text{the set $\{1,2,\ldots,n\}$,}\\
I^{(n)}\quad -&\quad \text{the identity matrix $n\times n$.}
\end{align*}
To be precise, let us remark that the term ``numbering'' with
respect to the edges connected in the graph vertex of the degree $n$
means strictly numbering by the elements of the set $\hat{n}$.

\begin{theorem}\label{volneOP}
Let us consider a quantum graph vertex $V$ of the degree $n$.
\begin{itemize}
\item[(i)] If $m\leq n$, $S\in\C^{m,m}$ is a self-adjoint matrix
and $T\in\C^{m,n-m}$, then the equation
\begin{equation}\label{ST}
\left(\begin{array}{cc}
I^{(m)} & T \\
0 & 0
\end{array}\right)\Psi'_V=
\left(\begin{array}{cc}
S & 0 \\
-T^* & I^{(n-m)}
\end{array}\right)\Psi_V
\end{equation}
expresses admissible boundary conditions. This statement holds
true for any numbering of the edges.
\item[(ii)] For any vertex coupling there exist a number $m\leq n$
and a numbering of edges such that the coupling is described by
the boundary conditions \eqref{ST} with the uniquely given
matrices $T\in\C^{m,n-m}$ and self-adjoint $S\in\C^{m,m}$.
\item[(iii)] Consider a quantum graph vertex of the degree $n$ with
the numbering of the edges explicitly given; then there is a
permutation $\Pi\in S_n$ such that the boundary conditions may be
written in the modified form
\begin{equation}\label{Coupling}
\left(\begin{array}{cc}
I^{(m)} & T \\
0 & 0
\end{array}\right)\tilde{\Psi}_V'=
\left(\begin{array}{cc}
S & 0 \\
-T^* & I^{(n-m)}
\end{array}\right)\tilde{\Psi}_V
\end{equation}
for
\begin{equation*}
\tilde{\Psi}_V=\left(\begin{array}{c}
\psi_{\Pi(1)}(0)\\
\vdots\\
\psi_{\Pi(n)}(0)
\end{array}\right)\,,
\qquad \tilde{\Psi}'_V=\left(\begin{array}{c}
\psi_{\Pi(1)}'(0)\\
\vdots\\
\psi_{\Pi(n)}'(0)
\end{array}\right)\,,
\end{equation*}
where the self-adjoint matrix $S\in\C^{m,m}$ and the matrix
$T\in\C^{m,n-m}$ depend unambiguously on $\Pi$. This formulation
of boundary conditions is in general not unique, since there may
be different admissible permutations $\Pi$, but one can make it
unique by choosing the lexicographically smallest permutation
$\Pi$.
\end{itemize}
\end{theorem}
\begin{proof}
The claim (iii) is an immediate consequence of (ii) using a
simultaneous permutation of elements in the vectors $\Psi_V$ and
$\Psi'_V$, so we have to prove the first two. As for (i), we have
to show that the vertex coupling~\eqref{1} with matrices
$$
A=\left(\begin{array}{cc}
-S & 0 \\
T^* & -I^{(n-m)}
\end{array}\right)
\quad\text{and}\quad
B=\left(\begin{array}{cc}
I^{(n)} & T \\
0 & 0
\end{array}\right)\,,
$$
conform with \eqref{KS}. We have
$$
\mathrm{rank}\left(\begin{array}{cccc}
-S & 0 & I^{(m)} & T \\
T^* & -I^{(n-m)} & 0 & 0
\end{array}\right)=
\mathrm{rank}\left(\begin{array}{cccc}
I^{(m)} & 0 & -S & T \\
0 & -I^{(n-m)} & T^* & 0
\end{array}\right)=n
$$
and
$$
\left(\begin{array}{cc}
-S & 0 \\
T^* & -I^{(n-m)}
\end{array}\right)\cdot
\left(\begin{array}{cc}
I^{(n)} & T \\
0 & 0
\end{array}\right)^*=
\left(\begin{array}{cc}
-S & 0 \\
0 & 0
\end{array}\right)\,;
$$
the latter matrix is self-adjoint since $S=S^*$, thus \eqref{KS}
is satisfied.

Now we proceed to (ii). Consider a quantum graph vertex of the degree
$n$ with an arbitrary fixed vertex coupling. Let $\Psi_V$ and
$\Psi'_V$ denote the vectors of values and derivatives of the wave
function components at the edge ends; the order of the components
is arbitrary but fixed and the same for both vectors. We know that
the coupling can be described by boundary conditions (\ref{1})
with some $A,B\in\C^{n,n}$ satisfying~\eqref{KS}. Our aim is to
find a number $m\leq n$, a certain numbering of the edges and
matrices $S$ and $T$ such that the boundary conditions~\eqref{1}
are equivalent to~\eqref{ST}. Moreover, we have to show that such
a number $m$ is the only possible and that $S,T$ depend uniquely
on the edge numbering.

When proceeding from \eqref{1} to \eqref{ST}, we may use
exclusively manipulations that do not affect the meaning of the
coupling, namely
\begin{itemize}
\setlength{\itemsep}{-3pt}
\item simultaneous permutation of
columns of the matrices $A,B$ combined with corresponding
simultaneous permutation of components in $\Psi_V$ and $\Psi'_V$,
\item multiplying the system from left by a regular matrix.
\end{itemize}
We see from \eqref{ST} that $m$ is equal to the rank of the matrix
applied at $\Psi'_V$. We observe that the rank of this matrix, as
well as of that applied at $\Psi_V$, is not influences by any
of the manipulations mentioned above, hence it is obvious that
$m=\mathrm{rank}(B)$ and that such a choice is the only possible,
i.e. $m$ is unique.

Since $\mathrm{rank}(B)=m$ with $m\in\{0,\ldots,n\}$, there is an
$m$-tuple of linearly independent columns of the matrix $B$;
suppose that their indices are $j_1,\ldots,j_m$. We permute
simultaneously the columns of $B$ and $A$ so that those with
indices $j_1,\ldots,j_m$ are now at the positions $1,\ldots,m$,
and the same we do with the components of the vectors $\Psi_V$,
$\Psi'_V$. Labelling the permuted matrices $A,B$ and vectors
$\Psi_V$, $\Psi'_V$ with tildes, we get
\begin{equation}\label{1vln}
\tilde{A}\tilde{\Psi}_V+\tilde{B}\tilde{\Psi}'_V=0\,.
\end{equation}
Since $\rank(\tilde{B})=\rank(B)=m$, there are $m$ rows of
$\tilde{B}$ that are linearly independent, let their indices be
$i_1,\ldots,i_m$, and $n-m$ rows that are linear combinations of
the preceding ones. First we permute the rows in~\eqref{1vln} so
that those with indices $i_1,\ldots,i_m$ are put to the positions
$1,\ldots,m$; note that it corresponds to a matrix multiplication
of the whole system~\eqref{1vln} by a permutation matrix (which is
regular) from the left, i.e. an authorized manipulation. In this
way we pass from $\tilde{A}$ and $\tilde{B}$ to matrices which we
denote as $\check{A}$ and $\check{B}$; it is obvious that this
operation keeps the first $m$ columns of the matrix $\check{B}$
linearly independent.

In the next step we add to each of the last $n-m$ rows of
$\check{A}\tilde{\Psi}(0)+\check{B}\tilde{\Psi}'(0)=0$ such a
linear combination of the first $m$ rows that all the last $n-m$
rows of $\check{B}$ vanish. This is possible, because the last
$n-m$ lines of $\check{B}$ are linearly dependent on the first $m$
lines. It is easy to see that it is an authorized operation, not
changing the meaning of the boundary conditions; the resulting
matrices at the LHS will be denoted as $\hat{B}$ and $\hat{A}$,
i.e.
\begin{equation}\label{1hat}
\hat{A}\tilde{\Psi}_V+\hat{B}\tilde{\Psi}'_V=0\,.
\end{equation}
From the construction described above we know that the matrix
$\hat{B}$ has a block form,
\begin{equation*}
\hat{B}=\left(\begin{array}{cc}
\hat{\mathcal{B}}_{11} & \hat{\mathcal{B}}_{12} \\
0 & 0
\end{array}\right)\,,
\end{equation*}
where $\hat{\mathcal{B}}_{11}\in\C^{m,m}$ and $\hat{\mathcal{B}}_{12}
\in\C^{m,n-m}$; the square matrix $\hat{\mathcal{B}}_{11}\in\C^{m,m}$ is
regular, because its columns are linearly independent. We proceed
by multiplying the system~\eqref{1hat} from the left by the matrix
$$
\left(\begin{array}{cc}
\hat{\mathcal{B}}_{11}^{-1} & 0 \\
0 & I^{(n-m)}
\end{array}\right)\,,
$$
arriving at boundary conditions
\begin{equation}\label{Vazba}
\left(\begin{array}{cc}
\mathcal{A}_{11} & \mathcal{A}_{12} \\
\mathcal{A}_{21} & \mathcal{A}_{22}
\end{array}\right)\tilde{\Psi}_V+
\left(\begin{array}{cc}
I^{(m)} & \mathcal{B}_{12} \\
0 & 0
\end{array}\right)\tilde{\Psi}'_V=0\,,
\end{equation}
where $\mathcal{B}_{12}=\hat{\mathcal{B}}_{11}^{-1}\hat{\mathcal{B}}_{12}$.

Boundary conditions~\eqref{Vazba} are equivalent to \eqref{1},
therefore they have to be admissible. In other words, the matrices
$\left(\begin{array}{cc}
\mathcal{A}_{11} & \mathcal{A}_{12} \\
\mathcal{A}_{21} & \mathcal{A}_{22}
\end{array}\right)$
and
$\left(\begin{array}{cc}
I^{(m)} & \mathcal{B}_{12} \\
0 & 0
\end{array}\right)$
have to satisfy both the conditions~\eqref{KS}, which we are now
going to verify. Let us begin with the second one. We have
\begin{equation*}
\left(\begin{array}{cc}
\mathcal{A}_{11} & \mathcal{A}_{12} \\
\mathcal{A}_{21} & \mathcal{A}_{22}
\end{array}\right)\cdot
\left(\begin{array}{cc}
I^{(m)} & 0 \\
\mathcal{B}_{12}^* & 0
\end{array}\right)=
\left(\begin{array}{cc}
\mathcal{A}_{11}+\mathcal{A}_{12}\mathcal{B}_{12}^* & 0 \\
\mathcal{A}_{21}+\mathcal{A}_{22}\mathcal{B}_{12}^* & 0
\end{array}\right)
\end{equation*}
and this matrix is self-adjoint if and only if
$\mathcal{A}_{11}+\mathcal{A}_{12}\mathcal{B}_{12}^*$ is self
adjoint and $\mathcal{A}_{21}+\mathcal{A}_{22}
\mathcal{B}_{12}^*=0$. We infer that $\mathcal{A}_{21}
=-\mathcal{A}_{22}\mathcal{B}_{12}^*$, hence
condition~\eqref{Vazba} acquires the form
\begin{equation}\label{Vazba1}
\left(\begin{array}{cc}
\mathcal{A}_{11} & \mathcal{A}_{12} \\
-\mathcal{A}_{22}\mathcal{B}_{12}^* & \mathcal{A}_{22}
\end{array}\right)\tilde{\Psi}_V+
\left(\begin{array}{cc}
I^{(m)} & \mathcal{B}_{12} \\
0 & 0
\end{array}\right)\tilde{\Psi}'_V=0\,.
\end{equation}
The first one of the conditions~\eqref{KS} says that
$$
\rank\left(\begin{array}{cccc}
\mathcal{A}_{11} & \mathcal{A}_{12} & I^{(m)} & \mathcal{B}_{12} \\
-\mathcal{A}_{22}\mathcal{B}_{12}^* & \mathcal{A}_{22} & 0 & 0
\end{array}\right)=n\,,
$$
hence $\rank\left(-\mathcal{A}_{22}\mathcal{B}_{12}^* |
\mathcal{A}_{22}\right)=n-m$. Since
$\left(-\mathcal{A}_{22}\mathcal{B}_{12}^* |
\mathcal{A}_{22}\right)=-\mathcal{A}_{22}\cdot\left(\mathcal{B}_{12}^*|
I^{(n-m)}\right)$ we obtain the condition $\rank(\mathcal{A}_{22})
=n-m$, i.e. $\mathcal{A}_{22}$ must be a regular matrix. It allows
us to multiply the equation~\eqref{Vazba1} from the left by the
matrix
$$
\left(\begin{array}{cc}
I^{(m)} & -\mathcal{A}_{12}\mathcal{A}_{22}^{-1} \\
0 & -\mathcal{A}_{22}^{-1}
\end{array}\right)\,,
$$
which is obviously well-defined and regular; this operation leads
to the condition
\begin{equation*}
\left(\begin{array}{cc}
\mathcal{A}_{11}+\mathcal{A}_{12}\mathcal{B}_{12}^* & 0 \\
\mathcal{B}_{12}^* & -I^{(n-m)}
\end{array}\right)\tilde{\Psi}_V+
\left(\begin{array}{cc}
I^{(m)} & \mathcal{B}_{12} \\
0 & 0
\end{array}\right)\tilde{\Psi}'_V=0\,.
\end{equation*}
If follows from our previous considerations that the square matrix
$\mathcal{A}_{11}+\mathcal{A}_{12}\mathcal{B}_{12}^*$ is
self-adjoint. If we denote it as $-S$, rename the block
$\mathcal{B}_{12}$ as $T$ and transfer the term containing
$\tilde{\Psi}'_V$ to the right hand side, we arrive at boundary
conditions
\begin{equation}\label{Vazba2}
\left(\begin{array}{cc}
I^{(m)} & T \\
0 & 0
\end{array}\right)\tilde{\Psi}'_V=
\left(\begin{array}{cc}
S & 0 \\
-T^* & I^{(n-m)}
\end{array}\right)\tilde{\Psi}_V\,.
\end{equation}
The order of components in $\tilde{\Psi}_V$ and $\tilde{\Psi}'_V$
determines just the appropriate numbering, in other words, the
vectors $\tilde{\Psi}_V$ and $\tilde{\Psi}'_V$ represent exactly
what we understood by $\Psi_V$ and $\Psi'_V$ in the formulation of
the theorem.

Finally, the uniqueness of the matrices $S$ and $T$ with respect
to the choice of the permutation $\Pi$ is a consequence of the
presence of the blocks $I^{(m)}$ and $I^{(n-m)}$. First of all,
the block $I^{(n-m)}$ implies that there is only one possible $T$,
otherwise the conditions for $\tilde{\psi}'_{m+1},\ldots,
\tilde{\psi}'_{n}$ would change, and next, the block $I^{(m)}$
together with the uniqueness of $T$ implies that there is only one
possible $S$, otherwise the conditions for $\tilde{\psi}_{1},
\ldots,\tilde{\psi}_{m}$ would change.
\end{proof}

\begin{rem}\label{redukcepar}
The expression~\eqref{Coupling} implies, in particular, that if
$B$ has not full rank, the number of real numbers parametrizing
the vertex coupling~\eqref{1} is reduced from $n^2$ to at most
$m(2n-m)=n^2-(n-m)^2$, where $m=\mathrm{rank}(B)$. Another
reduction can come from a lower rank of the matrix $A$.
\end{rem}
\begin{rem}
The procedure of permuting columns and applying linear
transformations to the rows of the system \eqref{1} has been done
with respect to the matrix $B$, but one can start by same right
from the matrix $A$ as well. In this way we would obtain similar
boundary conditions as \eqref{ST}, only the vectors $\Psi_V$ and
$\Psi'_V$ would be interchanged. Theorem \ref{volneOP} can be thus formulated with Equation \eqref{ST} replaced by
$$
\left(\begin{array}{cc}
I^{(m)} & T \\
0 & 0
\end{array}\right)\Psi_V=
\left(\begin{array}{cc}
S & 0 \\
-T^* & I^{(n-m)}
\end{array}\right)\Psi'_V\,.
$$
For completeness' sake we add that another possible forms of Equation \eqref{ST} in Theorem \ref{volneOP} are
$$
\left(\begin{array}{cc}
S & 0 \\
-T^* & I^{(n-m)}
\end{array}\right)\Psi_V+
\left(\begin{array}{cc}
I^{(m)} & T \\
0 & 0
\end{array}\right)\Psi'_V=0
$$
and
$$
\left(\begin{array}{cc}
I^{(m)} & T \\
0 & 0
\end{array}\right)\Psi_V+
\left(\begin{array}{cc}
S & 0 \\
-T^* & I^{(n-m)}
\end{array}\right)\Psi'_V=0\,;
$$
having the standardized form $A\Psi_V+B\Psi'_V=0$, last two formulations may be sometimes more convenient than \eqref{ST}.\\
Obviously, an analogous remark applies to Equation \eqref{Coupling}.
\end{rem}
\begin{rem}
A formulation of boundary conditions with a matrix structure
singling out the regular part as in \eqref{Coupling} has been
derived in a different way by P. Kuchment \cite{Ku04}. Recall that in the setting
analogous to ours he stated existence of an orthogonal projector
$P$ in $\C^n$ with the complementary projector $Q=Id-P$ and a
self-adjoint operator $L$ in $Q\C^n$
such that the boundary conditions may be written in the form
\begin{equation}\label{Kuchment}
\begin{array}{c}
P\Psi_V=0\\
Q\Psi'_V+LQ\Psi_V=0\,.
\end{array}
\end{equation}
Let us briefly explain how P. Kuchment's form differs from \eqref{Coupling}. When transformed into a matrix form, \eqref{Kuchment} consists of two groups of $n$ linearly dependent equations. If we then naturally extract a single group of $n$ linearly indepent ones, we arrive at a condition with a structure similar to \eqref{Vazba1}, i. e. the upper right submatrix standing at $\Psi'_V$ is generally a \emph{nonzero} matrix $m\times(n-m)$. In other words, whilst P. Kuchment aimed to decompose the boundary conditions with respect to two complementary orthogonal projectors, our aim was to obtain a unique matrix form with as many vanishing terms as possible; the form \eqref{ST} turned out to have a highly suitable structure for solving the problem of approximations that we are going to analyze in the rest of the paper.
\end{rem}

To conclude this introductory section, let us summarize main
advantages and disadvantages of the conditions \eqref{ST} and
\eqref{Coupling}. They are unique and exhibit a simple and clear
correspondence between the parameters of the coupling and the
entries of matrices in \eqref{ST}, furthermore, the matrices in
\eqref{ST} are relatively sparse. On the negative side, the
structure of matrices in \eqref{ST} depends on $\mathrm{rank}(B)$
and the vertex numbering is not fully permutable.

\section{The approximation arrangement}

We have argued above that due to a local character one can
consider a single-vertex situation, i.e. star graph, when asking
about the meaning of the vertex coupling. In this section we
consider such a quantum graph with general matching conditions and
show that the singular coupling may be understood as a limit case
of certain family of graphs constructed only from edges connected
by $\delta$-couplings, $\delta$-interactions, and supporting
constant vector potentials.

Following the above discussion, one may consider the boundary
conditions of the form \eqref{ST}, renaming the edges if
necessary. It turns out that for notational purposes it is
advantageous to adopt the following convention on a shift of the
column indices of $T$:
\begin{conv}\label{indT}
The lines of the matrix $T$ are indexed from 1 to $m$, the columns
are indexed from $m+1$ to $n$.
\end{conv}

Now we can proceed to the description of our approximating model.
Consider a star graph with $n$ outgoing edges coupled in a general
way given by the condition~\eqref{Coupling}. The approximation in
question looks as follows (cf.~Fig.1):


\begin{itemize}
\item We take $n$ halflines, each parametrized by
$x\in[0,+\infty)$, with the endpoints denoted as $V_j$, and put a
$\delta$-coupling (to the edges specified below) with the
parameter $v_j(d)$ at the point $V_j$ for all $j\in\hat{n}$.
\item Certain pairs $V_j,V_k$ of halfline
endpoints will be joined by edges of the length $2d$, and the
center of each such joining segment will be denoted as
$W_{\{j,k\}}$. For each pair $\{j,k\}$, the points $V_j$ and
$V_k$, $j\neq k$, are joined if one of the following three
conditions is satisfied (keep in mind Convention~\ref{indT}):
\begin{itemize}
\item[(1)] $j\in\hat{m}$, $k\geq m+1$, and $T_{jk}\neq0$
(or $j\geq m+1$, $k\in\hat{m}$, and $T_{kj}\neq0$),
\item[(2)] $j,k\in\hat{m}$ and
$(\exists l\geq m+1)(T_{jl}\neq0\wedge T_{kl}\neq0)$,
\item[(3)] $j,k\in\hat{m}$, $S_{jk}\neq0$, and the previous
condition is not satisfied.
\end{itemize}
\item At each point $W_{\{j,k\}}$ we place a $\delta$ interaction
with a parameter $w_{\{j,k\}}(d)$. From now on we use the
following convention: the connecting edges of the length $2d$ are
considered as composed of two line segments of the length $d$, on
each of them the variable runs from 0 (corresponding to the point
$W_{\{j,k\}}$) to $d$ (corresponding to the point $V_j$ or $V_k$).
\item On each connecting segment described above we put a vector
potential which is constant on the whole line between the points
$V_j$ and $V_k$. We denote the potential strength between the
points $W_{\{j,k\}}$ and $V_j$ as $A_{(j,k)}(d)$, and between the
points $W_{\{j,k\}}$ and $V_k$ as $A_{(k,j)}(d)$. It follows from
the continuity that $A_{(k,j)}(d)=-A_{(j,k)}(d)$ for any pair
$\{j,k\}$.
\end{itemize}
The choice of the dependence of $v_j(d)$, $w_{\{j,k\}}(d)$ and
$A_{(j,k)}(d)$ on the parameter $d$ is crucial for the
approximation and will be specified later.
\begin{figure}[!t]
\begin{center}
\includegraphics[width=5cm, keepaspectratio]{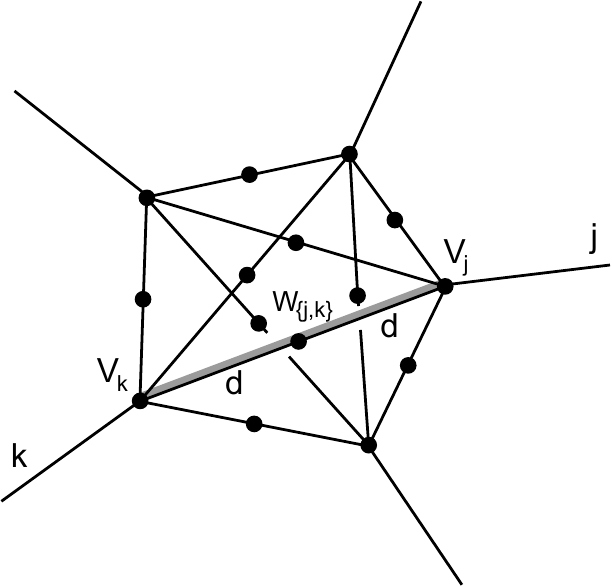}
\end{center}
\caption{The scheme of the approximation.  All inner links are of length $2d$.
Some connection links may be missing if the conditions given in the text 
are not satisfied. 
The quantities corresponding to the index pair $\{j,k\}$ are marked, and
the grey line symbolizes the vector potential $A_{(j,k)}(d)$.}
\end{figure}

It is useful to introduce the set $N_j\subset\hat{n}$ containing
indices of all the edges that are joined to the $j$-th one by a
connecting segment, i.e.
\begin{align}\label{Nj}
N_j=&\{k\in\hat{m}|\, S_{jk}\neq0\}\cup\{k\in\hat{m}|\,
  (\exists l\geq m+1)(T_{jl}\neq0\wedge T_{kl}\neq0)\} \nonumber \\
  &\cup\{k\geq m+1|\, T_{jk}\neq0\} \qquad \text{for } j\in\hat{m}\\
N_j=&\{k\in\hat{m}|\, T_{kj}\neq0\} \qquad\qquad\quad \text{for }
j\geq m+1 \nonumber
\end{align}
The definition of $N_j$ has these two trivial consequences, namely
\begin{gather}
k\in N_j\Leftrightarrow j\in N_k \label{Nj1} \\
j\geq m+1\Rightarrow N_j\subset\hat{m} \label{Nj2}
\end{gather}
\noindent For the wave function components on the edges we use the
following symbols:
\begin{itemize}
\item the wave function on the $j$-th half line is denoted by
$\psi_j$,
\item the wave function on the line  connecting points
$V_j$ and $V_k$ has two components: the one on the line between
$W_{\{j,k\}}$ and $V_j$ is denoted by $\varphi_{(j,k)}$, the one
on the half between the middle and the endpoint of the $k$-th half
line is denoted by $\varphi_{(k,j)}$. We remind once more the way
in which the variable $x$ of $\varphi_{(j,k)}$ and
$\varphi_{(k,j)}$ is considered: it grows from 0 at the point
$W_{\{j,k\}}$ to $d$ at the point $V_j$ or $V_k$, respectively.
\end{itemize}
Next we describe how the $\delta$ couplings involved look like;
for simplicity we will refrain from indicating in the boundary
conditions the dependence of the parameters $u,\, v_j,\,
w_{\{j,k\}}$ on the distance $d$.

The $\delta$ interaction at the edge connecting the $j$-th and
$k$-th half line (of course, for $j,k\in\hat{n}$ such that $k\in
N_j$ only) is expressed through the conditions
\begin{equation}\label{I.serie}
\begin{array}{c}
\varphi_{(j,k)}(0)=\varphi_{(k,j)}(0)=:\varphi_{\{j,k\}}(0)\,,
\\ \vspace{.5em}
\varphi_{(j,k)}'(0_+)+\varphi_{(k,j)}'(0_+)
=w_{\{j,k\}}\varphi_{\{j,k\}}(0)\,,
\end{array}
\end{equation}
the $\delta$ coupling at the endpoint of the $j$-th half line ($j\in\hat{n}$) means
\begin{equation}\label{II.serie}
\begin{array}{c}
\psi_j(0)=\varphi_{(j,k)}(d)\quad \text{for all}\; k\in N_j\,,
\\ \vspace{.5em}
\psi_j'(0)-\sum_{k\in
N_j}\varphi_{(j,k)}'(d)
=v_j\psi_j(0)\,.
\end{array}
\end{equation}
Further relations which will help us to find the parameter
dependence on $d$ come from Taylor expansion. Consider first the
case without any added potential,
\begin{eqnarray}\label{bez pot.}
&& \begin{aligned}
&\varphi_{(j,k)}(d)
=\varphi_{\{j,k\}}(0)+d\,
\varphi_{(j,k)}'(0)+\OO(d^2)\,,\\
&\varphi_{(j,k)}'(d)
=\varphi_{(j,k)}'(0)+\OO(d)\,, \quad j,k\in\hat{n}\,.
\end{aligned}
\end{eqnarray}

To take the effect of added vector potentials into account, the
following lemma will prove useful:
\begin{lem}\label{potencial}
Let us consider a line parametrized by the variable $x\in(0,L)$,
$L\in(0,+\infty)\cup\{+\infty\}$, and let $H$ denote a Hamiltonian
of a particle on this line interacting with a potential $V$,
\begin{equation}\label{H bez v.p.}
H=-\frac{\d^2}{\d x^2}+V\,,
\end{equation}
sufficiently regular to make $H$ self-adjoint. We denote by
$\psi^{s,t}$ the solution of $H\psi=k^2\psi$ with the boundary
values $\psi^{s,t}(0)=s$, ${\psi^{s,t}}'(0)=t$. Consider the same
system with a vector potential $A$ added, again sufficiently
regular; the Hamiltonian is consequently given by
\begin{equation}\label{H s v.p.}
H_A=\left(-\i\frac{\d}{\d x}-A\right)^2+V\,.
\end{equation}
Let $\psi_A^{s,t}$ denote the solution of $H_A\psi=k^2\psi$ with
the same boundary values as before, i.e. $\psi_A^{s,t}(0)=s$,
${\psi_A^{s,t}}'(0)=t$. Then the function $\psi_A^{s,t}$ can be
expressed as
$$
\psi_A^{s,t}(x)=\e^{\i\int_0^x A(z)\d z}\cdot\psi^{s,t}(x)\qquad
\text{for all}\quad x\in(0,L)\,.
$$
\end{lem}
\begin{proof}
The aim is to prove that
$$
-{\psi^{s,t}}''+V\psi^{s,t}=k^2\psi^{s,t} \quad\wedge\quad \psi^{s,t}(0)=s \quad\wedge\quad {\psi^{s,t}}'(0)=t
$$
implies
$$
\left(-\i\frac{\d}{\d x}-A\right)^2\left(\e^{\i\int_0^x A(z)\d
z}\cdot\psi^{s,t}\right)+V\cdot\e^{\i\int_0^x A(z)\d
z}\cdot\psi^{s,t}=k^2\e^{\i\int_0^x A(z)\d z}\cdot\psi^{s,t}
$$
and $\psi_A^{s,t}(0)=s$, ${\psi_A^{s,t}}'(0)=t$. The last part is
obvious, since the exponential factor involved is equal to one, hence it
suffices to prove the displayed relation. It is straightforward
that the Hamiltonian $H_A$ acts generally as
$$
H_A=-\frac{\d^2}{\d x^2}+2\i A\frac{\d}{\d x}+\i A'+A^2+V\,.
$$
We substitute $\e^{\i\int_0^x A(z)\d z}\cdot\psi^{s,t}$ for $\psi$,
obtaining
\begin{multline*}
\left[H_A\left(\e^{\i\int_0^x A(z)\d z}\cdot\psi^{s,t}\right)\right](x)
=-\frac{\d^2}{\d x^2}\left(\e^{\i\int_0^x A(z)\d z}\cdot\psi^{s,t}\right)(x)+\\
+2\i A(x)\frac{\d}{\d x}\left(\e^{\i\int_0^x A(z)\d z}\cdot\psi^{s,t}\right)(x)
+\left(\i A'(x)+A(x)^2+V(x)\right)\e^{\i\int_0^x A(z)\d z}\cdot\psi^{s,t}(x)\,.
\end{multline*}
Now we express the derivatives applying the formula $\frac{\d}{\d
x}\int_0^x A(z)\d z=A(x)$. Most of the terms then cancel, it
remains only
$$
\left[H_A\left(\e^{\i\int_0^x A(z)\d z}\cdot\psi^{s,t}\right)\right](x)=
\e^{\i\int_0^x A(z)\d z}\cdot
\left(-{\psi^{s,t}}''(x)+V(x)\cdot\psi^{s,t}(x)\right)\,.
$$
Due to the assumption $-{\psi^{s,t}}''+V\psi^{s,t}=k^2\psi^{s,t}$,
we have
$$
\left[H_A\left(\e^{\i\int_0^x A(z)\d z}\cdot\psi^{s,t}\right)\right](x)
=k^2\e^{\i\int_0^x A(z)\d z}\cdot\psi^{s,t}(x)\,,
$$
what we have set out to prove.
\end{proof}

The lemma says that adding a vector potential on an edge of a
quantum graph has a very simple effect of changing the phase of
the wave function by the value $\int_0^x A(z)\d z$. We will work
in this paper with the special case of constant vector potentials
on the connecting segments of the lengths $2d$, hence the phase
shift will be given here as a product of the value $A$ and the
length in question.

Lemma~\ref{potencial} has the following very useful consequence.
\begin{coro}\label{Transformovana delta}
Consider a quantum graph vertex with $n$ outgoing edges indexed by
$1,\ldots,n$ and parametrized by $x\in(0,L_j)$. Suppose that there
is a $\delta$ coupling with the parameter $\alpha$ at the vertex,
and moreover, that there is a constant vector potential $A_j$ on
the $j$-th edge for all $j\in\hat{n}$. Let $\psi_j$ denote the
wave function component on the $j$-the edge. Then the boundary
conditions acquire the form
\begin{gather}
\psi_j(0)=\psi_k(0)=:\psi(0) \qquad \text{for all} \quad j,k\in\hat{n}\,, \\
\sum_{j=1}^n\psi_j'(0)=\left(\alpha+\i\sum_{j=1}^n
A_j\right)\psi(0)\,,
\end{gather}
where $\psi_j(0)$, $\psi_j'(0)$, etc., stand for the one-sided
(right) limits at $x=0$.
\end{coro}
\noindent In other words the effect of the vector potentials on
the boundary conditions corresponding to a ``pure'' $\delta$
coupling is the following:
\begin{itemize}
\item the continuity is not affected,
\item the coupling parameter is changed from $\alpha$ to
$\alpha+\i\sum_{j=1}^n A_j$.
\end{itemize}
\begin{proof}
Consider first the situation without any vector potentials. If
$\psi^0_j$, $j\in\hat{n}$, denote the wave function components
corresponding to this case, the boundary conditions expressing the
$\delta$ coupling have the form~\eqref{delta}, i.e.
\begin{equation}\label{delta1}
\begin{array}{c}
\psi^0_j(0)=\psi^0_k(0)=:\psi^0(0) \qquad \text{for all} \quad j,k\in\hat{n}\,,
\\ [.5em]
\sum_{j=1}^n{\psi^0_j}'(0)=\alpha\psi^0(0)\,.
\end{array}
\end{equation}
If there are vector potentials on the edges, $A_j$ on the $j$-th
edge, one has in view of the previous lemma, $\psi_j(x)=\e^{\i A_j
x}\psi^0_j(x)$, i.e.
\begin{gather*}
\psi^0_j(x)=\e^{-\i A_j x}\psi_j(x)\,,\\
{\psi^0_j}'(x)=\frac{\d}{\d x}\left(\e^{-\i A_j x}\psi_j(x)\right)
=\e^{-\i A_j x}{\psi_j}'(x)-\i A_j\cdot\e^{-\i A_j x}\psi_j(x)\,.
\end{gather*}
Thence we express $\psi^0_j(0)$ and ${\psi^0_j}'(0)$: they are
equal to
\begin{gather*}
\psi^0_j(0)=\psi_j(0)\,,\\
{\psi^0_j}'(0)={\psi_j}'(0)-\i A_j\psi_j(0)\,;
\end{gather*}
substituting them to~\eqref{delta1} we obtain
\begin{gather*}
\psi_j(0)=\psi_k(0)=:\psi(0) \qquad \text{for all} \quad j,k\in\hat{n}\,,\\
\sum_{j=1}^n\left(\psi_j'(0)-\i A_j\cdot\psi_j(0)\right)=\alpha\psi(0)\,.
\end{gather*}
The first line expresses the continuity of the wavefunction in the
vertex supporting the $\delta$ coupling in the same way as in the
absence of vector potentials, whereas the second line shows how
the condition for the sum of the derivatives is changed. With the
continuity in mind, we may replace $\psi_j(0)$ by $\psi(0)$
obtaining
$$
\sum_{j=1}^n\psi_j'(0)=\left(\alpha+\i\sum_{j=1}^n A_j\right)\psi(0)\,,
$$
which finishes the proof.
\end{proof}

Recall that approximating we are constructing supposes that
constant vector potentials are added on the joining edges. If an
edge of the length $2d$ joins the endpoints of the $j$-th and
$k$-th half line, there is a constant vector potential of the
value $A_{(j,k)}(d)$ on the part of the length $d$ closer to the
$j$-th half line and a constant vector potential of the value
$A_{(k,j)}(d) = -A_{(j,k)}(d)$ on the part of the length $d$
closer to the $k$-th half line. With regard to
Lemma~\ref{potencial}, the impact of the added potentials consists
in phase shifts by $d\cdot A_{(j,k)}(d)$ and $d\cdot
A_{(k,j)}(d)$. Let us include this effect into the corresponding
equations, i.e. into~\eqref{bez pot.}:
\begin{eqnarray}\label{III.serie}
&& \begin{aligned} &\varphi_{(j,k)}(d) =\e^{\i dA_{(j,k)}}
(\varphi_{\{j,k\}}(0)+d\,
\varphi_{(j,k)}'(0))+\OO(d^2)\,,\\
&\varphi_{(j,k)}'(d)
=\e^{\i dA_{(j,k)}}\varphi_{(j,k)}'(0)+\OO(d)\,, \quad j,k\in\hat{n}\,.
\label{IIIa.serie}
\end{aligned}
\end{eqnarray}
The system of equations~\eqref{I.serie}, \eqref{II.serie},
and~\eqref{III.serie} describes the relations between values of
wave functions and their derivatives at all the vertices. Next we
will eliminate the terms with the ``auxiliary'' functions
$\varphi_{\{j,k\}}$ and express the relations between $2n$ terms
$\psi_j(0)$, $\psi_j'(0)$, $j\in\hat{n}$.

We begin with the first one of the relations~\eqref{IIIa.serie}
together with the continuity requirement~\eqref{II.serie}, which
yields
\begin{equation}\label{hvezdicka}
d\: \varphi_{(j,k)}'(0)=\e^{-\i dA_{(j,k)}}\psi_j(0)
-\varphi_{\{j,k\}}(0)+\OO(d^2)\,.
\end{equation}
The same relation holds with $j$ replaced by $k$, summing them
together and using the second of the relations~\eqref{I.serie} we
get
$$
\left(2+d\: w_{\{j,k\}}\right)\varphi_{\{j,k\}}(0)
=\e^{-\i dA_{(j,k)}}\psi_j(0)+\e^{-\i dA_{(k,j)}}\psi_k(0)+\OO(d^2)\,.
$$
We express $\varphi_{\{j,k\}}(0)$ from here and substitute into
the first of the equations~\eqref{IIIa.serie}; using at the same
time the first relation of~\eqref{II.serie} we get
$$
\psi_j(0)=\e^{\i dA_{(j,k)}}\cdot \left(\frac{\e^{-\i
dA_{(j,k)}}\psi_j(0)+\e^{-\i dA_{(k,j)}}\psi_k(0)+\OO(d^2)}
{2+d\cdot w_{\{j,k\}}}+d\: \varphi_{(j,k)}'(0)\right)+\OO(d^2)\,,
$$
and considering the second of the equations~\eqref{IIIa.serie}, we
have
$$
\psi_j(0)=\frac{\psi_j(0)+\e^{\i d(A_{(j,k)}-A_{(k,j)})}\psi_k(0)+\OO(d^2)}
{2+d\cdot w_{\{j,k\}}}+d\: \varphi_{(j,k)}'(d)+\OO(d^2)\,.
$$
Since the values of vector potentials are supposed to have the
``antisymmetry'' property, $A_{(k,j)}(d)=-A_{(j,k)}(d)$, we may
simplify the above equation to
\begin{equation}\label{phi'0+}
\psi_j(0)=\frac{\psi_j(0)+\e^{2\i dA_{(j,k)}}\psi_k(0)+\OO(d^2)}
{2+d\cdot w_{\{j,k\}}}+d\: \varphi_{(j,k)}'(d)+\OO(d^2)\,.
\end{equation}
Summing the last equation over $k\in N_j$ yields
\begin{multline*}
\#N_j\cdot\psi_j(0)=\psi_j(0)\cdot\sum_{k\in N_j}\frac{1}{2+d\cdot w_{\{j,k\}}}
+\sum_{k\in N_j}\frac{\e^{2\i dA_{(j,k)}}\psi_k(0)}
{2+d\cdot w_{\{j,k\}}}
+d\cdot\sum_{k\in N_j}\varphi_{(j,k)}'(d)+\\
+\sum_{k\in N_j}\frac{\O(d^2)}{2+d\cdot w_{\{j,k\}}}+\OO(d^2)
\end{multline*}
($\#N_j$ denotes the cardinality of $N_j$), and with the help of the second of the relations~\eqref{II.serie}
we arrive at the final expression,
\begin{multline}\label{j-ty radek}
d\psi_j'(0)=\left(d v_j+\#N_j-\sum_{k\in N_j}\frac{1}{2+d\cdot w_{\{j,k\}}}\right)\psi_j(0)
-\sum_{k\in N_j}\frac{\e^{\i d(A_{(j,k)}-A_{(k,j)})}\psi_k(0)}
{2+d\cdot w_{\{j,k\}}}\\
+\sum_{k\in N_j}\frac{\O(d^2)}{2+d\cdot w_{\{j,k\}}}+\OO(d^2)\,.
\end{multline}

Our objective is to choose $v_j(d)$, $w_{\{j,k\}}(d)$ and
$A_{(j,k)}(d)$ in such a way that in the limit $d\to0$ the system
of relations~\eqref{j-ty radek} with $j\in\hat{n}$ tends to the
system of $n$ boundary conditions~\eqref{Coupling}. The lines
of~\eqref{Coupling} are of two types, let us recall:
\begin{align}
\psi_j'(0)+\sum_{l=m+1}^n T_{jl}\psi'_l(0)&
=\sum_{k=1}^m S_{jk}\psi_k(0) & j\in\hat{m}\,\, \label{j<=m} \\
0&=-\sum_{k=1}^m \overline{T_{kj}}\psi_k(0)+\psi_j(0)
& j=m+1,\ldots,n \label{j>=m+1} \,.
\end{align}
We point out here with reference to~\eqref{Nj} that these
relations may be written also with the summation indices running
through the restricted sets, namely
\begin{align}
\psi_j'(0)+\sum_{l\in N_j\backslash\hat{m}} T_{jl}\psi'_l(0)&
=\sum_{k=1}^m S_{jk}\psi_k(0) & j\in\hat{m}\,\, \label{j<=m.} \\
0&=-\sum_{k\in N_j} \overline{T_{kj}}\psi_k(0)+\psi_j(0)
& j=m+1,\ldots,n \label{j>=m+1.} \,,
\end{align}
since for any pair $j\in\hat{m}$, $l\in\{m+1,\cdots,n\}$ the
implication $T_{jl}\neq0\Rightarrow l\in N_j$ holds, see also
Eqs.~\eqref{Nj1}, \eqref{Nj2}.

When looking for a suitable dependence of $v_j(d)$,
$w_{\{j,k\}}(d)$ and $A_{(j,k)}(d)$ on $d$, we start with
Eq.~\eqref{j-ty radek} in the case when $j\geq m+1$. Our aim is to
find conditions under which~\eqref{j-ty radek} tends
to~\eqref{j>=m+1.} as $d\to0$. It is obvious that the sufficient
conditions are
\begin{eqnarray}
\lim_{d\to0}\left(d v_j+\#N_j-\sum_{k\in N_j}
\frac{1}{2+d\cdot w_{\{j,k\}}}\right)\in\R\backslash\{0\}\,, \label{(1)}\\
\lim_{d\to0}\frac{1}{2+d\cdot w_{\{j,k\}}}\in\R\backslash\{0\}\quad
\forall k\in N_j\,, \label{(2)}\\
\frac{\frac{\e^{2\i dA_{(j,k)}}}
{2+d\cdot w_{\{j,k\}}}}{d v_j+\#N_j-\sum_{h\in N_j}\frac{1}{2+d\cdot w_{\{j,h\}}}}
=\overline{T_{kj}}\quad \forall k\in N_j\,. \label{(3)}
\end{eqnarray}
Now we proceed to the case $j\in\hat{m}$. Our approach is based on
substitution of~\eqref{j-ty radek} into the left-hand side
of~\eqref{j<=m.} and a subsequent comparison of the right-hand
sides. The substitution is straightforward,
\begin{multline}\label{j-ty radek1}
\psi_j'(0)+\sum_{l\in N_j\backslash\hat{m}} T_{jl}\cdot\psi_l'(0)=
\left(v_j+\frac{\#N_j}{d}-\frac{1}{d}\sum_{h\in N_j}
\frac{1}{2+d\cdot w_{\{j,h\}}}\right)\psi_j(0)
-\frac{1}{d}\sum_{k\in N_j}
\frac{\e^{2\i dA_{(j,k)}}\psi_k(0)}{2+d\cdot w_{\{j,k\}}}\\
+\sum_{l\in N_j\backslash\hat{m}}
T_{jl}\left[\left(v_l+\frac{\#N_l}{d} -\frac{1}{d}\sum_{h\in
N_l}\frac{1}{2+d\cdot w_{\{l,h\}}}\right)\psi_l(0)-
\frac{1}{d}\sum_{k\in N_l}\frac{\e^{2\i dA_{(l,k)}}\psi_k(0)}
{2+d\cdot w_{\{l,k\}}}\right] \\
+\OO(d)+\sum_{k\in N_j}\frac{\O(d)}{2+d\cdot w_{\{j,k\}}}
+\sum_{l=m+1}^n T_{jl}\left(\O(d) +\sum_{h\in
N_l}\frac{\O(d)}{2+d\cdot w_{\{l,h\}}}\right)\,,
\end{multline}
then we apply two identities, which can be easily proven, namely
\begin{equation*}
(i)\quad \sum_{k\in N_j}\frac{\e^{2\i dA_{(j,k)}}\psi_k(0)}{2+d\cdot w_{\{j,k\}}}=
\sum_{k\in N_j\cap\hat{m}}\frac{\e^{2\i dA_{(j,k)}}\psi_k(0)}{2+d\cdot w_{\{j,k\}}}+
\sum_{l\in N_j\backslash\hat{m}}\frac{\e^{2\i dA_{(j,l)}}\psi_l(0)}{2+d\cdot w_{\{j,l\}}}\,,
\qquad\qquad\qquad\qquad\quad\ \ 
\end{equation*}
\begin{multline*}
(ii)\quad \sum_{l\in N_j\backslash\hat{m}} T_{jl}\sum_{k\in N_l}
\frac{\e^{2\i dA_{{(l,k)}}}\psi_k(0)}{2+d\cdot w_{\{l,k\}}}\\
=\left(\sum_{l\in N_j\backslash\hat{m}} T_{jl}
\frac{\e^{2\i dA_{(l,j)}}}{2+d\cdot w_{\{l,j\}}}\right)\psi_j(0)
+\sum_{k\in N_j\cap\hat{m}}\left(\sum_{l\in N_k\backslash\hat{m}}
T_{jl}\frac{\e^{2\i dA_{(l,k)}}}{2+d\cdot w_{\{l,k\}}}\right)\psi_k(0)
\end{multline*}
and obtain
\begin{multline}\label{j-ty radek11}
\psi_j'(0)+\sum_{l\in N_j\backslash\hat{m}} T_{jl}\cdot\psi_l'(0)
\\ =\left(v_j+\frac{\#N_j}{d}-\frac{1}{d}\sum_{h\in N_j}
\frac{1}{2+d\cdot w_{\{j,h\}}}
-\frac{1}{d}\sum_{l\in N_j\backslash\hat{m}}T_{jl}\frac{\e^{2\i dA_{(l,j)}}}{2+d\cdot w_{\{l,j\}}}\right)\psi_j(0)\\
-\frac{1}{d}\sum_{k\in N_j\cap\hat{m}}\left(\frac{\e^{2\i dA_{(j,k)}}}{2+d\cdot w_{\{j,k\}}}
+\sum_{l\in N_k\backslash\hat{m}} T_{jl}\frac{\e^{2\i dA_{(l,k)}}}{2+d\cdot w_{\{l,k\}}}\right)\psi_k(0)\\
+\sum_{l\in N_j\backslash\hat{m}}\left(-\frac{1}{d}\cdot\frac{\e^{2\i dA_{(j,l)}}}{2+d\cdot w_{\{j,l\}}}
+T_{jl}\left(v_l+\frac{\#N_l}{d}-\frac{1}{d}\sum_{h\in N_l}\frac{1}{2+d\cdot w_{\{l,h\}}}\right)\right)\psi_l(0)\\
+\OO(d)+\sum_{k\in N_j}\frac{\O(d)}{2+d\cdot w_{\{j,k\}}}
+\sum_{k\in N_j\cap\hat{m}}\sum_{l\in N_k\cap N_j\backslash\hat{m}}\frac{\O(d)}{2+d\cdot w_{\{l,k\}}}\,.
\end{multline}
As we have announced above, the goal is to determine terms
$v_j(d)$, $w_{\{j,k\}}(d)$ and $A_{(j,k)}(d)$ such that if
$d\to0$, the right-hand side of~\eqref{j-ty radek11} tends to the
eight-hand side of~\eqref{j<=m.} for all $j\in\hat{m}$. We observe
that this will be the case provided
\begin{gather}
v_j+\frac{\#N_j}{d}-\frac{1}{d}\sum_{h\in N_j}
\frac{1}{2+d\cdot w_{\{j,h\}}}
-\frac{1}{d}\sum_{l\in N_j}T_{jl}\frac{\e^{2\i dA_{(l,j)}}}{2+d\cdot w_{\{l,j\}}}=S_{jj}\,, \label{(4)}\\
-\frac{1}{d}\frac{\e^{2\i dA_{(j,k)}}}{2+d\cdot w_{\{j,k\}}}
-\frac{1}{d}\sum_{l\in N_k\backslash\hat{m}} T_{jl}
\frac{\e^{2\i dA_{(l,k)}}}{2+d\cdot w_{\{l,k\}}}=S_{jk}
\quad \forall k\in N_j\cap\hat{m}\,, \label{(5)}\\
-\frac{1}{d}\frac{\e^{2\i dA_{(j,l)}}}{2+d\cdot w_{\{j,l\}}}
+T_{jl}\left(v_l+\frac{\#N_l}{d}-\frac{1}{d}\sum_{h\in N_l}
\frac{1}{2+d\cdot w_{\{l,h\}}}\right)=0\quad \forall l\in N_j\backslash\hat{m}\,, \label{(6)}\\
\lim_{d\to0}\frac{1}{2+d\cdot w_{\{j,k\}}}\in\R \quad
\forall k\in N_j\,, \label{(7)}\\
\lim_{d\to0}\frac{1}{2+d\cdot w_{\{l,k\}}}\in\R \quad
\forall k\in N_j\cap\hat{m},\,l\in N_k\cap N_j\backslash\hat{m}\,. \label{(8)}
\end{gather}
It is easy to see that the set of equations~\eqref{(6)} for
$j\in\hat{m}$, $l\in N_j\backslash\hat{m}$ is equivalent to the
set~\eqref{(3)} for $j\geq m+1$, $k\in N_j$. Similarly,
Eq.~\eqref{(8)} for $j\in\hat{m}$, $k\in N_j\cap\hat{m}$, $l\in
N_k\cap N_j\backslash\hat{m}$ is a weaker set of conditions
than~\eqref{(2)} with $j\geq m+1$, $k\in N_j$. Finally,
Eq.~\eqref{(7)} reduces for $k\in N_j\backslash\hat{m}$
to~\eqref{(2)}, thus it suffices to consider~\eqref{(7)} with
$k\in N_j\cap\hat{m}$.

The procedure of determination of $v_j(d)$, $w_{\{j,k\}}(d)$ and
$A_{(j,k)}(d)$ will proceed in three steps, at the end we add the
fourth step involving the verification of the limit
conditions~\eqref{(1)}, \eqref{(2)}, and~\eqref{(7)} restricted to
$k\in N_j\cap\hat{m}$.

\medskip

\noindent \textit{Step I.} We use Eq.~\eqref{(6)} to find an
expression for $w_{\{j,l\}}(d)$ and $A_{(j,l)}(d)$ when
$j\in\hat{m}$ and $l\in N_j\backslash\hat{m}$. We begin with
rearranging Eq.~\eqref{(3)} into the form
\begin{equation}\label{ww}
\frac{1}{2+d\cdot w_{\{j,l\}}}
=\e^{-2\i dA_{(j,l)}}\cdot T_{jl}\left(d v_l+\#N_l
-\sum_{h\in N_l}\frac{1}{2+d\cdot w_{\{l,h\}}}\right)
\quad \forall l\in N_j\backslash\hat{m}\,.
\end{equation}
Since all the terms except $\e^{-2\i dA_{(j,l)}}$ and $T_{jl}$ are
real, we can obtain immediately a condition for $A_{(j,l)}$: We
put
\begin{equation*}
\e^{2\i dA_{(j,l)}}=\left\{\begin{array}{ccl} T_{jl}/\|T_{jl}\| &
\text{if} & \Re\, T_{jl}\geq0\,, \\ [.5em] -T_{jl}/\|T_{jl}\| &
\text{if} & \Re\, T_{jl}<0\,;
\end{array}\right.
\end{equation*}
it is easy to see that such a choice ensures that the expression
$\e^{-2\i dA_{(j,l)}}\cdot T_{jl}$ is always real. The vector
potential strength may be then chosen as follows,
\begin{equation}\label{A m,n-m}
A_{(j,l)}(d)=\left\{\begin{array}{lcl} \frac{1}{2d}\arg\,T_{jl} &
\text{if} & \Re\, T_{jl}\geq0\,, \\
[.5em] \frac{1}{2d}\left(\arg\,T_{jl}-\pi\right) & \text{if} &
\Re\, T_{jl}<0
\end{array}\right.
\end{equation}
for all $j\in\hat{m}$, $l\in N_j\backslash\hat{m}$. We remark that
this choice is obviously not the only one possible. Note that in
this situation, namely if $j\in\hat{m}$ and $l\in
N_j\backslash\hat{m}$, the potentials do not depend on $d$.
Taking~\eqref{A m,n-m} into account, Eq.~\eqref{ww} simplifies to
\begin{equation}\label{www}
\frac{1}{2+d\cdot w_{\{j,l\}}}=\langle T_{jl}\rangle \cdot\left(d
v_l+\#N_l-\sum_{h\in N_l} \frac{1}{2+d\cdot w_{\{l,h\}}}\right)
\quad \forall l\geq m+1, j\in N_l\,;
\end{equation}
note that $j\in\hat{m}\wedge l\in N_j\backslash\hat{m}
\Leftrightarrow l\geq m+1\wedge j\in N_l$. The symbol
$\langle\cdot\rangle$ here has the following meaning: if $c\in\C$,
then
\begin{equation*}
\langle c\rangle=\left\{\begin{array}{ccl}
|c| & \text{if} & \Re\, c\geq0\,, \\
-|c| & \text{if} & \Re\, c<0\,.
\end{array}\right.
\end{equation*}
Summing \eqref{www} over $j\in N_l$ we get
$$
\sum_{j\in N_l}\frac{1}{2+d\cdot w_{\{j,l\}}}
=\sum_{j\in N_l}\langle T_{jl}\rangle\cdot\left(d v_l+\#N_l
-\sum_{h\in N_l}\frac{1}{2+d\cdot w_{\{l,h\}}}\right)\,,
$$
i.e.
$$
\left(1+\sum_{h\in N_l}\langle T_{hl}\rangle\right)
\sum_{j\in N_l}\frac{1}{2+d\cdot w_{\{j,l\}}}
=\sum_{j\in N_l}\langle T_{jl}\rangle\cdot\left(d v_l+\#N_l\right)\,.
$$
We have to distinguish here two situations:\\
\textit{(i)} If $1+\sum_{h\in N_l}\langle T_{hl}\rangle\neq0$, one
obtains
$$
\sum_{h\in N_l}\frac{1}{2+d\cdot w_{\{l,h\}}}
=\frac{\sum_{h\in N_l}\langle T_{hl}
\rangle}{1+\sum_{h\in N_l}\langle T_{hl}\rangle}\cdot(d v_l+\#N_l)\,,
$$
and the substitution of the left-hand side into the right-hand
side of~\eqref{www} leads to the formula for $1/(2+d\cdot
w_{\{j,l\}})$, namely
$$
\frac{1}{2+d\cdot w_{\{j,l\}}}
=\langle T_{jl}\rangle\cdot\frac{d v_l+\#N_l}{1+\sum_{h\in N_l}
\langle T_{hl}\rangle} \qquad \forall j\in\hat{m},\ l\in N_j\backslash\hat{m}\,.
$$
We observe that the sum in the denominator may be computed over
the whole set $\hat{m}$ as well, since $h\notin\N_l\Rightarrow
T_{hl}=0$, which slightly simplifies the formula,
\begin{equation*}\label{ww m,n-m}
\frac{1}{2+d\cdot w_{\{j,l\}}}=\langle T_{jl}
\rangle\cdot\frac{d v_l+\#N_l}{1+\sum_{h=1}^m\langle T_{hl}\rangle}
\qquad \forall j\in\hat{m},\ l\in N_j\backslash\hat{m}\,.
\end{equation*}
From here one can easily express $w_{\{j,l\}}$, if $v_l$ is known.
However, it turns out that $v_l(d)$, $l\geq m+1$ can be chosen
almost arbitrarily, the only requirements are to keep the
expression $d v_l+\#N_l$ nonzero and to satisfy~\eqref{(1)},
\eqref{(2)} and \eqref{(7)}. The simplest choice possible is to
define $v_l$ by the expression
$$
\frac{d v_l+\#N_l}{1+\sum_{h=1}^m\langle T_{hl}\rangle}=1\,,
$$
which simplifies the expressions for other parameters. Here we
obtain already expressions for $v_l$ and $w_{\{j,l\}}$ if $l\geq
m+1$, viz
\begin{equation}\label{v n-m}
v_l(d)=\frac{1-\#N_l+\sum_{h=1}^m\langle T_{hl}\rangle}{d}
\qquad \forall l\geq m+1\,,
\end{equation}
\begin{equation}\label{w m,n-m}
w_{\{j,l\}}(d)=\frac{1}{d}\left(-2+\frac{1}{\langle T_{jl}\rangle}\right)
\qquad \forall j\in\hat{m},\ l\in N_j\backslash\hat{m}\,.
\end{equation}
\textit{(ii)} If $1+\sum_{h\in N_l}\langle T_{hl}\rangle=0$, it
holds necessarily $d v_l+\#N_l=0$, and consequently,
$v_l=-\frac{\#N_l}{d}$. Note that this equation may be obtained
from Eq.~\eqref{v n-m} by putting formally $1+\sum_{h=1}^m\langle
T_{hl}\rangle=0$. It is easy to check that $w_{\{j,l\}}$ given by
Eq.~\eqref{w m,n-m} satisfies~\eqref{ww} in the case $1+\sum_{h\in
N_l}\langle T_{hl}\rangle=0$ as well. Summing these facts up, we
conclude that Eqs. \eqref{v n-m}, \eqref{w m,n-m} hold universally
regardless whether $1+\sum_{h\in N_l}\langle T_{hl}\rangle$ equals
zero or not.

We would like to stress that the freedom in the choice of $v_l(d)$
is a consequence of the fact mentioned in Remark~\ref{redukcepar},
namely that the number of parameters of a vertex coupling
decreases with the decreasing value of $\mathrm{rank}(B)$.

\medskip

\noindent \textit{Step II.} Equation~\eqref{(5)} together with the
results of Step I will be used to determine $w_{\{j,k\}}(d)$ and
$A_{(j,k)}(d)$ in the case when $j\in\hat{m}$ and $k\in
N_j\cap\hat{m}$. From~\eqref{(5)} we have
$$
-\frac{\e^{2\i dA_{(j,k)}}}{2+d\cdot w_{\{j,k\}}}
=d\cdot S_{jk}+\sum_{l\in N_k\backslash\hat{m}}
T_{jl}\frac{\e^{2\i dA_{(l,k)}}}{2+d\cdot w_{\{l,k\}}}\,;
$$
the pairs $(l,k)$ appearing in the sum are of the type examined in
Step I, i.e. $k\in\hat{m}$, $l\in N_k\backslash\hat{m}$). Thus one
may substitute from \eqref{v n-m} and \eqref{w m,n-m} to obtain
$$
-\frac{\e^{2\i dA_{(j,k)}}}{2+d\cdot w_{\{j,k\}}}
=d\cdot S_{jk}+\sum_{l\in N_k\backslash\hat{m}} T_{jl}\overline{T_{kl}}\,.
$$
We observe that the summation index may run through the whole set
$\hat{n}\backslash\hat{m}$, because $l\geq m+1\wedge l\notin
N_k\Rightarrow T_{kl}=0$. This allows one to obtain a more elegant
formula. In a similar way as above, we find the expression for
$A_{(j,k)}$,
\begin{subequations}\label{A m,m}
\begin{equation}
A_{(j,k)}(d)=\frac{1}{2d}\arg\,\left(d\cdot S_{jk}
+\sum_{l=m+1}^n T_{jl}\overline{T_{kl}}\right) \qquad
\text{for}\quad \Re\left(d\cdot S_{jk}
+\sum_{l=m+1}^n T_{jl}\overline{T_{kl}}\right)\geq0
\end{equation}
and
\begin{equation}
A_{(j,k)}(d)=\frac{1}{2d}\left[\arg\,\left(d\cdot S_{jk}
+\sum_{l=m+1}^n T_{jl}\overline{T_{kl}}\right)-\pi\right] \quad
\text{for}\quad \Re\left(d\cdot S_{jk}
+\sum_{l=m+1}^n T_{jl}\overline{T_{kl}}\right)<0\,,
\end{equation}
\end{subequations}
and for $w_{\{j,k\}}$,
\begin{equation}\label{w m,m}
\frac{1}{2+d\cdot w_{\{j,k\}}}=-\left\langle d\cdot S_{jk}
+\sum_{l=m+1}^n T_{jl}\overline{T_{kl}}\right\rangle\,.
\end{equation}

\medskip

\noindent \textit{Step III.} Substitution of the results of Steps
I and II into Eq.~\eqref{(4)} provides an expression for $v_j(d)$
in the case when $j\in\hat{m}$. A simple calculation gives
\begin{equation}\label{v m}
v_j(d)=S_{jj}-\frac{\#N_j}{d}-\sum_{k=1}^m\left\langle
S_{jk}+\frac{1}{d}\sum_{l=m+1}^n
T_{jl}\overline{T_{kl}}\right\rangle
+\frac{1}{d}\sum_{l=m+1}^n(1+\langle T_{jl}\rangle)\langle
T_{jl}\rangle\,.
\end{equation}
Since $S$ is a self-adjoint matrix, the term $S_{jj}$ is real,
thus the whole right-hand side is a real expression.

\vspace*{12pt}

\noindent \textit{Step IV.} Finally, we verify
conditions~\eqref{(1)}, \eqref{(2)}, and~\eqref{(7)}, the last one
being restricted to $k\in N_j\cap\hat{m}$. This step consists in
trivial substitutions:
\begin{align*}
\eqref{(1)}:\quad & \lim_{d\to0}\left(d v_j+\#N_j-\sum_{k\in N_j}
\frac{1}{2+d\cdot w_{\{j,k\}}}\right)=\lim_{d\to0}1=1
\in\R\backslash\{0\}\quad \forall j\geq m+1\,,\\
\eqref{(2)}:\quad & \lim_{d\to0}\frac{1}{2+d\cdot w_{\{j,k\}}}
=\lim_{d\to0}\langle T_{kj}\rangle=\langle T_{kj}
\rangle\in\R\backslash\{0\}\quad \forall j\geq m+1,\, k\in N_j\,,\\
\begin{split}
\eqref{(7)}:\quad &
\lim_{d\to0}\frac{1}{2+d\cdot w_{\{j,k\}}}
=-\lim_{d\to0}\left\langle d\cdot S_{jk}+\sum_{l=m+1}^n T_{jl}\overline{T_{kl}}\right\rangle
=\left\langle\sum_{l=m+1}^n T_{jl}\overline{T_{kl}}\right\rangle\in\R\\
&\forall j\in\hat{m},\, k\in N_j\cap\hat{m}\,.
\end{split}
\end{align*}



\section{The norm-resolvent convergence}

In the previous section we have shown that any vertex coupling in
the center point of a star graph may be regarded as a limit of a
certain family of graphs supporting nothing but $\delta$
couplings, $\delta$ interactions and constant vector potentials.
The parameter values of all the $\delta$'s and vector potentials
have been derived using a method devised originally in \cite{CS98,
SMMC99} for the case of a generalized point interaction on the
line. The aim of this section is to give a clear meaning to this
convergence. Specifically, we are going to show that the
Hamiltonian of the approximating system converges to the
Hamiltonian of the approximated system \emph{in the norm-resolvent
sense}, with the natural consequences for the convergence of
eigenvalues, eigenfunctions, etc.

We denote the Hamiltonian of the star graph $\Gamma$ with the
coupling (\ref{ST}) at the vertex as $H^\mathrm{Ad}$ (referring to
the \emph{a}pproximate\emph{d} system), and $H^\mathrm{Ag}_d$ will
stand for the \emph{a}pproximatin\emph{g} family of graphs that
has been constructed in the previous section. Symbols
$R^\mathrm{Ad}(k^2)$ and $R^\mathrm{Ag}_d(k^2)$ will then denote
the resolvents of $H^\mathrm{Ad}$ and $H^\mathrm{Ag}_d$ at the
points $k^2$ from the resolvent set. Needless to say, the
operators act on different spaces: $R^\mathrm{Ad}(k^2)$ on
$L^2(G)$, where $G=(\R^+)^n$ corresponds to the star graph
$\Gamma$, and $R^\mathrm{Ag}_d(k^2)$ on $L^2(G_d)$, where
\begin{equation}\label{Gd}
G_d=(\R^+)^n\oplus (0,d)^{\sum_{j=1}^n N_j}\,.
\end{equation}
Our goal is to compare these resolvents. In order to do that,
we need to identify $R^\mathrm{Ad}(k^2)$ with the orthogonal sum
\begin{equation}\label{Gdecomp}
R^\mathrm{Ad}_d(k^2)=R^\mathrm{Ad}(k^2)\oplus0\,,
\end{equation}
where $0$ is a zero operator acting on the space
$L^2\left((0,d)^{\sum_{j=1}^n N_j}\right)$ which is removed in the
limit. Then both the operators $R^\mathrm{Ad}_d(k^2)$ and
$R^\mathrm{Ag}_d(k^2)$ are defined as acting on functions from the
set $G_d$ which are vector functions with $n+\sum_{j=1}^n N_j$
components; we will index the components by the set
\begin{equation}\label{mathcal I}
\mathcal{I}=\hat{n}\cup\left\{\left.(l,h)\right|\,l\in\hat{n},h\in N_l\right\}\,.
\end{equation}
In this setting we are able to state now the main theorem of this
section and the whole paper.

\begin{theorem}
Let $v_j,\; j\in\hat{n}$, $,\,w_{\{j,k\}}\; j\in\hat{n}, k\in\N_j$
and $A^{(j,k)}(d)$ depend on $d$ according to \eqref{v m},
\eqref{v n-m}, \eqref{w m,m}, \eqref{w m,n-m}, \eqref{A m,m} and
\eqref{A m,n-m}, respectively.
Then the family $H^\mathrm{Ag}(d)$ converges to $H^\mathrm{Ad}_d$
in the norm-resolvent sense as $d\to 0_+$.
\end{theorem}

\begin{proof}
We have to compare the resolvents $R^\mathrm{Ad}_d(k^2)$ and
$R^\mathrm{Ag}_d(k^2)$. It is obviously sufficient to check the
convergence in the Hilbert-Schmidt norm,
$$
\left\|R^\mathrm{Ag}_d(k^2)-R^\mathrm{Ad}_d(k^2)\right\|_2\to0_+
\quad \text{ as }\: d\longrightarrow0_+\,,
$$
in other words, to show that the difference of the corresponding
resolvent kernels denoted as $\GG^{\mathrm{Ag},d}_k$ and
$\GG^{\mathrm{Ad},d}_k$, respectively, tends to zero in $L^2(G_d\oplus
G_d)$. Recall that these kernels, or Green's functions, are in our
case matrix functions with $\left(n+\sum_{j=1}^n
N_j\right)\times\left(n+\sum_{j=1}^n N_j\right)$ entries. We will
index the entries by pairs of indices taken from the set
$\mathcal{I}$ (cf.~\eqref{mathcal I}).

The proof is divided into three parts. In the first and the second
part we will derive the resolvent kernels $\GG^{\mathrm{Ag},d}_k$
and $\GG^{\mathrm{Ad},d}_k$, respectively, in the last part we compare them
and demonstrate the norm-resolvent convergence.

\medskip

\noindent \textit{I. Resolvent of the approximated Hamiltonian}

\medskip

\noindent Let us construct first $\GG^\mathrm{Ad}_k$ for the
star-graph Hamiltonian with the condition (\ref{1}) at the vertex.
We begin with $n$ independent halflines with Dirichlet condition
at its endpoints; Green's function for each of them is well-known
to be
$$
\GG_{\i\kappa}(x,y)=\frac{\sinh\kappa x_<\: \e^{-\kappa
x_>}}{\kappa}\,,
$$
where $x_<:=\min\{x,y\},\:x_>:=\max\{x,y\}$, and we put
$\i\kappa=k$ assuming conventionally $\Re\kappa>0$. The sought
Green's function is then given by Krein's formula
\cite[App.~A]{AGHH},
\begin{equation}\label{KreinAd}
R^\mathrm{Ad}(k^2)=R^{Hl}(k^2) +\sum^{n}_{j,l=1}\lambda_{jl}(k^2)
\left(\phi_l\left(\overline{k^2}\right),\cdot\right)_{
L^2((\mathbb{R}^+)^n)}\phi_j(k^2)\,,
\end{equation}
where $R^{Hl}(k^2)$ acts on each half line as an integral operator
with the kernel $\GG_{\i\kappa}$, and for $\phi_j(k^2)$ one can
choose any elements of the deficiency subspaces of the largest
common restriction; we will work with
$\left(\phi_j(k^2)(\vec{x})\right)_l =\delta_{jl}\e^{-\kappa x_j}$
where the symbol $\vec{x}$ stands here for the vector
$(x_1,\ldots,x_n)\in(\R^+)^n$. Then we have
\begin{equation*}
R^\mathrm{Ad}(k^2)\left(\begin{array}{c}
\psi_1\\
\vdots\\
\psi_n
\end{array}\right)
\left(\begin{array}{c}
x_1\\
\vdots\\
x_n
\end{array}\right)
=
\left(\begin{array}{c}
\int^{+\infty}_{0}\G_{\i\kappa}(x_1,y_1)\psi_1(y_1)\d y_1\\
\vdots\\
\int^{+\infty}_{0}\G_{\i\kappa}(x_n,y_n)\psi_n(y_n)\d y_n
\end{array}\right)
+
\sum^{n}_{j,l=1}\lambda_{jl}(k^2)\left(\e^{-\bar{\kappa}y_l},\psi_l(y_l)\right)_{L^2(\R^+)}
\left(\begin{array}{c}
0\\
\vdots\\
\e^{-\kappa x_j}\\
\vdots\\
0
\end{array}
\right)\,,
\end{equation*}
which should be satisfied for any
$\left(\psi_1,\ldots,\psi_n\right)^T\in \bigoplus_{j=1}^n
L^2(\mathbb{R}^+)$. We observe that for all $j\in\hat{n}$, the
$j$-th component on the right hand side depends only on the
variable $x_j$. That is why one can consider the components as
functions of one variable; we denote them as $g_j(x_j)$,
$j\in\hat{n}$, in other words,
$$
R^\mathrm{Ad}(k^2)\left(\begin{array}{c}
\psi_1\\
\vdots\\
\psi_n
\end{array}\right)
\left(\begin{array}{c}
x_1\\
\vdots\\
x_n
\end{array}\right)
=:
\left(\begin{array}{c}
g_1(x_1)\\
\vdots\\
g_n(x_n)
\end{array}\right)\,.
$$
The functions $g_j$ are therefore given explicitly by
$$
g_j(x_j)=\int^{+\infty}_{0}\GG_{\i\kappa}(x_j,y)\psi_j(y)\D
y+\sum^{n}_{l=1}\lambda_{jl}(k^2)\int^{+\infty}_{0}\e^{-\kappa
y} \psi_l(y)\D y\cdot\e^{-\kappa x_j}\,.
$$
Since the resolvent maps the whole Hilbert space into the domain
of the operator, these functions have to satisfy the boundary
conditions at the vertex,
\begin{equation}\label{KrPodm}
\sum^{n}_{h=1}A_{jh}g_h(0)+\sum^{n}_{h=1}B_{jh}g_h'(0)=0 \quad
\text{for all }\; j\in\hat{n}\,,
\end{equation}
where
$$
A=\left(\begin{array}{cc}
S & 0 \\
-T^* & I^{(n-m)}
\end{array}\right)\,,\quad
-B=\left(\begin{array}{cc}
I^{(m)} & T \\
0 & 0
\end{array}\right)\,.
$$
Using the explicit form of $\GG_{\i\kappa}(x_h,y)$ and the equality
$\left.\frac{\partial\GG_\kappa(x_h,y)}{\partial
x_h}\right|_{x_h=0}=\e^{-\kappa y}$, we find
\begin{equation*}
g_h(0)=\sum^{n}_{l=1}\lambda_{hl}(k^2)\int^{+\infty}_{0}\e^{-\kappa
y}\psi_l(y)\D y
\end{equation*}
and
\begin{equation*}
g_h'(0)=\int^{+\infty}_{0}\e^{-\kappa y}\psi_h(y)\D
y-\kappa\sum^{n}_{l=1}\lambda_{hl}(k^2)\int^{+\infty}_{0}\e^{-\kappa
y}\psi_l(y)\D y\,.
\end{equation*}
Substituting from these two relations into~\eqref{KrPodm} we get a
system of equations,
$$
\sum^{n}_{l=1}\int^{+\infty}_{0}
\left(\sum^{n}_{h=1}A_{jh}\lambda_{hl}(k^2)+B_{jl}
-\kappa\sum^{n}_{h=1}B_{jh}\lambda_{hl}(k^2)\right) \e^{-\kappa
y}\psi_l(y)\D y=0
$$
with $j\in\hat{n}$. We require that the left-hand side vanishes
for any $\psi_1,\psi_2,\ldots,\psi_n$; this yields the condition
$A\Lambda+B-\kappa B\Lambda=0$. From here it is easy to find the
matrix $\Lambda(k^2)$: we have $(A-\kappa B)
\Lambda=-B$, and therefore
$$
\Lambda(k^2)=(A-\kappa B)^{-1}(-B)\,.
$$
Substituting the
explicit forms of $A$ and $-B$ into the expression for $\Lambda$,
we obtain
$$
\Lambda(k^2)=
\left(\begin{array}{cc}
S+\kappa I^{(m)} & \kappa T \\
-T^* & I^{(n-m)}
\end{array}\right)^{-1}
\left(\begin{array}{cc}
I^{(m)} & T \\
0 & 0
\end{array}\right)\,,
$$
or explicitly
\begin{equation*}
\Lambda(k^2)=
\left(\begin{array}{cc}
(S+\kappa I^{(m)}+\kappa TT^*)^{-1} & (S+\kappa I^{(m)}+\kappa TT^*)^{-1}T \\
T^*(S+\kappa I^{(m)}+\kappa TT^*)^{-1} & T^*(S+\kappa I^{(m)}+\kappa TT^*)^{-1}T
\end{array}\right)
\end{equation*}
provided that $(S+\kappa I^{(m)}+\kappa TT^*)^{-1}$ is well
defined. To check that the matrix $S+\kappa I^{(m)}+\kappa TT^*$
is regular, we notice that
\begin{equation}\label{RegularitaMatice}
S+\kappa I^{(m)}+\kappa TT^*=S+\kappa(I^{(m)}+TT^*)\,,
\end{equation}
where the matrix $I^{(m)}+TT^*$ is positive definite and thus
regular, and the value $\kappa$ may be chosen arbitrarily with the
only restriction $\Re\kappa>0$. Consequently, it suffices to
choose $\Re\kappa$ big enough to make the matrix
$\kappa(I^{(m)}+TT^*)$ dominate over $S$, which ensures the
regularity of $S+\kappa(I^{(m)}+TT^*)$.

Having found the coefficients $\lambda_{jl}(k^2)$, we have fully
determined the Green's function $\GG^\mathrm{Ad}_{\i\kappa}$ of
the approximated system. Recall that it is an $n\times n$
matrix-valued function the ($j,l$)-th element of which is given by
\begin{equation}\label{G Ad}
\GG^\mathrm{Ad}_{\i\kappa,jl}(x,y)= \delta_{jl}\frac{\sinh\kappa
x_<\: \e^{-\kappa x_>}}{\kappa} +\lambda_{jl}(k^2)\ \e^{-\kappa
x}\e^{-\kappa y}\,;
\end{equation}
we use the convention that $x$ is from the $j$-th halfline and $y$
from the $l$-th one. The kernel of the operator
$R^\mathrm{Ad}_d(k^2)$ is according to \eqref{Gdecomp} given simply
by
\begin{equation}\label{G Ad,d}
\GG^{\mathrm{Ad},d}_{\i\kappa}=\left(\begin{array}{c|c}
\GG^\mathrm{Ad}_{\i\kappa} & 0 \\
\hline
0 & 0
\end{array}\right)\,,
\end{equation}
i.e. all entries of $\G^{\mathrm{Ad},d}_{\i\kappa}$ except for those indexed by $j,l\in\hat{n}$ vanish.

\medskip

\noindent \textit{II. Resolvents of the approximating family of
Hamiltonians}

\medskip

\noindent Next we will pass to resolvent construction for the
approximating family of operators $H^\mathrm{Ag}_d$. As a starting
point we consider $n$ independent halflines and $\sum_{j=1}^n N_j$
lines of the length $d$ with constant vector potentials
$A_{(j,l)}(d)$, both halflines and lines of the finite length are
supposed to have Dirichlet endpoints. We know that the Green's
function is $\GG_{\i\kappa}(x,y)=\kappa^{-1} \sinh\kappa x_<
\e^{-\kappa x_>}$ in the case of the halflines. The Green's
function in the case of the lines of the length $d$ will be found
in two steps. We begin with a line without vector potential and
with Dirichlet endpoints; the Green's function can be easily
derived being equal to
$$
\tilde{\GG}_{\i\kappa}(x,y)=\frac{\sinh\kappa x_< \sinh\kappa(d-x_>)}{\kappa\sinh\kappa d}\,.
$$
The Hamiltonian of a free particle on a line segment acts as
$-\frac{\d^2}{\d x^2}$, if a vector potential $A$ is added it
changes to $\left(-\i\frac{\d}{\d x}-A\right)^2$. Using
Lemma~\ref{Transformovana delta} it is easy to check that
\begin{equation}\label{transformace}
\left(-\i\frac{\d}{\d x}-A\right)^2=U\left(-\frac{\d^2}{\d x^2}\right)U^*\,,
\end{equation}
where $U$ is the unitary operator acting as
$$
(U\psi)(x)=\e^{\i A x}\psi(x)\,.
$$
If we denote $H_0=-\frac{\d^2}{\d x^2}$ and $H_A=\left(-\i\frac{\d}{\d x}-A\right)^2$, we see
that
$$
\left(H_A-\lambda\right)^{-1}=\left(UH_0U^*-\lambda\right)^{-1}
=\left(U(H_0-\lambda\right)U^*)^{-1}=U\left(H_0-\lambda\right)^{-1}U^*\,,
$$
so the corresponding resolvents are related by the relation
analogous to~\eqref{transformace}. This yields
\begin{multline*}
\left((H_A-\lambda)^{-1}\psi\right)(x)=\left(U(H_A-\lambda)^{-1}U^*\psi\right)(x)
=\e^{\i Ax}\int_0^d\tilde{\GG}_{\i\kappa}(x,y)\, \e^{-\i Ay}\psi(y)\d y\\
=\int_0^d\e^{\i Ax}\frac{\sinh\kappa x_<
\sinh\kappa(d-x_>)}{\kappa\sinh\kappa d}\, \e^{-\i Ay}\psi(y)\d
y\,,
\end{multline*}
thus the sought integral kernel is equal to
$$
\tilde{\GG}^A_{\i\kappa}(x,y)=\e^{\i Ax}\frac{\sinh\kappa x_<
\sinh\kappa(d-x_>)}{\kappa\sinh\kappa d}\, \e^{-\i Ay}\,.
$$

Now we can proceed to the derivation of the complete resolvent
$R^\mathrm{Ag}_d(k^2)$ which will be done again by means of the
Krein's formula. The situation here is more complicated than in
the case of the approximated system; recall that
$R^\mathrm{Ag}_d(k^2)$, as well as $H^\mathrm{Ad}_d(k^2)$, acts on the
larger Hilbert space $L^2(G_d)$, where $G_d$ has been defined
in~\eqref{Gd}.
Moreover, the application of Krein's formula means that we have to
connect all the line segments using the appropriate boundary
conditions, i.e. we must change boundary conditions at
$n+2\sum_{j=1}^n N_j$ endpoints, specifically $n$ belonging to $n$
half lines and $2\sum_{j=1}^n N_j$ belonging to $\sum_{j=1}^n N_j$
segments of the length $d$. Thus the index set for the indices in
the sum on the right hand side of the formula has $n+2\sum_{j=1}^n
N_j$ elements; we will index them by the set
$$
\hat{\mathcal{I}}=\hat{n}\cup\left\{\left.(l,h)^0\right|\,l\in\hat{n},h\in
N_l\right\} \cup\left\{\left.(l,h)^d\right|\,l\in\hat{n},h\in
N_l\right\}\,.
$$
The elements of $\hat{n}$ correspond to changed boundary
conditions at the endpoints of the half lines, and the elements of
the type $(l,h)^0$ and $(l,h)^d$ ($h\in N_l$) correspond to
changed boundary conditions at the endpoints of the segments of
the length $d$ which are connected to the endpoint of the $l$-th
half line.
If we denote by the symbol $R^\mathrm{Dc}_d(k^2)$ the resolvent of
the system of the $n+\sum_{j=1}^n N_j$ decomposed edges with
Dirichlet boundary conditions at the endpoints, Krein's formula
for this pair of operators has the form
\begin{equation}\label{KreinAg}
R^\mathrm{Ag}_d(k^2)=R^\mathrm{Dc}_d(k^2)+\sum_{J,L\in\hat{\mathcal{I}}}\lambda^d_{JL}(k^2)
\left(\phi^{d}_L\left(\overline{k^2}\right),\cdot\right)_{
L^2(G_d)}\phi^{d}_J(k^2)\,.
\end{equation}
The role of the superscript $d$ in the lambda symbols is to
distinguish them from $\lambda_{jl}$ that have been used in
Eq.~\eqref{KreinAd} for the resolvent of the approximated system.
The functions $\phi^{d}_J$ ($J\in\hat{\mathcal{I}}$) may be
chosen, as before in the case of the approximated system, as any
elements of the corresponding deficiency subspaces of the largest
common restriction. Note that each function $\phi^{d}_J$ has
$n+\sum_{j=1}^n N_j$ components indexed by elements of the set
$\mathcal{I}=\hat{n}\cup\left\{\left.(l,h)\right|\,l\in\hat{n},h\in
N_l\right\}$. It turns out that a suitable choice is
\begin{equation}\label{phi}
\begin{array}{rll}
\left(\phi_j(k^2)^d(\vec{x})\right)_{\tilde L}=&\delta_{j\tilde L}\e^{-\kappa x_j} & \text{for}\quad j\in\hat{n}\,,\tilde L\in\mathcal{I}\,, \\
\left(\phi_{(l,h)^0}^d(k^2)(\vec{x})\right)_{\tilde L}=&\e^{\i A_{(l,h)}x_{(l,h)}}\delta_{(l,h)\tilde L}\sinh\kappa x_{(l,h)} & \text{for}\quad l\in\hat{n},\,h\in N_l\,,\tilde L\in\mathcal{I} \\
\left(\phi_{(l,h)^d}^d(k^2)(\vec{x})\right)_{\tilde L}=&\e^{\i A_{(l,h)}x_{(l,h)}}\delta_{(l,h)\tilde L}\sinh\kappa(d-x_{(l,h)}) & \text{for}\quad l\in\hat{n},\,h\in N_l\,,\tilde L\in\mathcal{I}\,,
\end{array}
\end{equation}
where the symbol $\vec{x}$ denotes the vector from $G_d$ with the
components indexed by $\mathcal{I}$. We remark that if
$J\in\hat{n}$, $\phi^{d}_J$ is independent of $d$ and equal to the
corresponding function chosen above in the case of the
approximated system.

If we apply the operator~\eqref{KreinAg} to an arbitrary $\Psi\in
\bigoplus_{j=1}^n L^2(G_d)$, we obtain a vector function with
$n+\sum_{j=1}^n N_j$ components indexed by $\mathcal{I}$, we
denote them by $g_j$ ($j\in\hat{n}$) and $g_{(l,h)}$ with
$l\in\hat{n}, h\in N_l$. As in the case of the approximated
system, a component $g_J$ depends on $x_J$ only, thus each $g_J$
can be considered as a function of a single variable. A
calculation leads to the following explicit expressions for $g_j,
j\in\hat{n}$ and $g_{(l,h)}, l\in\hat{n}, h\in N_l$; for better
clarity we distinguish the integral variables on $\R^+$ and on
$(0,d)$ by a tilde, i.e. $y\in\R^+$, $\tilde{y}\in(0,d)$.
\begin{subequations}\label{g aproximujici}
\begin{multline}
g_j(x_j)=\int_{0}^{+\infty}\GG_{\i\kappa}(x_j,y)\psi_j(y)\,\D
y+\sum_{j'=1}^{n}\lambda^d_{jj'}(k^2)\int_{0}^{+\infty}\e^{-\kappa
y}\cdot\psi_{j'}(y)\,\D y\cdot\e^{-\kappa x_j}\\
+\sum_{l'=1}^{n}\sum_{h'\in N_{l'}}\left(\lambda^d_{j(l'h')^0}(k^2)\int_{0}^{d}\e^{-\i A_{(l',h')}\tilde{y}}\sinh\kappa
\tilde{y}\cdot\psi_{(l',h')}(\tilde{y})\,\D\tilde{y}\right.\\
+\left.\lambda^d_{j(l'h')^d}(k^2)\int_{0}^{d}\e^{-\i
A_{(l',h')}\tilde{y}}\sinh\kappa
(d-\tilde{y})\cdot\psi_{(l',h')}(\tilde{y})\,\D\tilde{y}\right)
\cdot\e^{-\kappa x_j} \,.
\end{multline}
\begin{multline}
g_{(l,h)}(x_{(l,h)})=\int_{0}^{d}\tilde{\GG}^{A_{(l,h)}}_{\i\kappa}(x_{(l,h)},\tilde{y})\psi_{(l,h)}(\tilde{y})\,\D \tilde{y}\\
+\e^{\i A_{(l,h)}x_{(l,h)}}\cdot\sinh\kappa x_{(l,h)}\cdot\left[\sum_{j'=1}^{n}\lambda^d_{(l,h)^0j'}(k^2)\int_{0}^{+\infty}\e^{-\kappa
y}\cdot\psi_{j'}(y)\,\D y \right.\\
+\sum_{l'=1}^{n}\sum_{h'\in N_{l'}}\left(\lambda^d_{(l,h)^0(l'h')^0}(k^2)\int_{0}^{d}\e^{-\i A_{(l',h')}\tilde{y}}\sinh\kappa
\tilde{y}\cdot\psi_{(l',h')}(\tilde{y})\,\D\tilde{y}\right.\\
+\left.\left.\lambda^d_{(l,h)^0(l'h')^d}(k^2)\int_{0}^{d}\e^{-\i
A_{(l',h')}\tilde{y}}\sinh\kappa
(d-\tilde{y})\cdot\psi_{(l',h')}(\tilde{y})\,\D\tilde{y}\right)
\right]\\
+\e^{\i A_{(l,h)}x_{(l,h)}}\cdot\sinh\kappa(d-x_{(l,h)})\cdot\left[\sum_{j'=1}^{n}\lambda^d_{(l,h)^dj'}(k^2)\int_{0}^{+\infty}\e^{-\kappa
y}\cdot\psi_{j'}(y)\,\D y \right.\\
+\sum_{l'=1}^{n}\sum_{h'\in N_{l'}}\left(\lambda^d_{(l,h)^d(l'h')^0}(k^2)\int_{0}^{d}\e^{-\i A_{(l',h')}\tilde{y}}\sinh\kappa
\tilde{y}\cdot\psi_{(l',h')}(\tilde{y})\,\D\tilde{y}\right.\\
+\left.\left.\lambda^d_{(l,h)^d(l'h')^d}(k^2)\int_{0}^{d}\e^{-\i
A_{(l',h')}\tilde{y}}\sinh\kappa
(d-\tilde{y})\cdot\psi_{(l',h')}(\tilde{y})\,\D\tilde{y}\right)
\right] \,.
\end{multline}
\end{subequations}
By definition the function $(g_J)_{J\in\mathcal{I}}$ belongs to
the domain of the operator $H^\mathrm{Ag}_{d}$, in particular, it
has to satisfy the boundary conditions at the points where the
edges are connected by $\delta$ interactions and $\delta$
couplings. Step by step we will write down now all these boundary
conditions; this will lead to the explicit expressions for the
coefficients $\lambda_{JL}^d(k^2)$.

\medskip

\noindent \textit{Step 1.} The continuity at the points
$W_{\{j,k\}}$ means
\begin{equation}\label{Krok1}
g_{(l,h)}(0)=g_{(h,l)}(0)
\end{equation}
for all $l\in\hat{n}, h\in N_l$. Since
$\tilde{\GG}^{A_{(l,h)}}_{\i\kappa}(0,\tilde{y})=0$ for all
$\tilde{y}\in(0,d)$, it holds
\begin{multline*}
g_{(l,h)}(0)=\sinh\kappa d\cdot\left[\sum_{j'=1}^{n}\lambda^d_{(l,h)^dj'}(k^2)\int_{0}^{+\infty}\e^{-\kappa
y}\psi_{j'}(y)\,\D y \right.\\
+\sum_{l'=1}^{n}\sum_{h'\in N_{l'}}\left(\lambda^d_{(l,h)^d(l'h')^0}(k^2)\int_{0}^{d}\e^{-\i A_{(l',h')}\tilde{y}}\sinh\kappa
\tilde{y}\psi_{(l',h')}(\tilde{y})\,\D\tilde{y}\right. \\
+\left.\left.\lambda^d_{(l,h)^d(l'h')^d}(k^2)\int_{0}^{d}\e^{-\i
A_{(l',h')}\tilde{y}}\sinh\kappa
(d-\tilde{y})\psi_{(l',h')}(\tilde{y})\,\D\tilde{y}\right)
\right]\,,
\end{multline*}
the expression for $g_{(h,l)}(0)$ is similar, just the positions
of $l$ and $h$ are interchanged. Since Eq.~\eqref{Krok1} must be
satisfied for any choice of the function
$\Psi=(\psi_J)_{J\in\mathcal{I}}$, the following equalities
obviously hold for $\forall l\in\hat{n}, h\in N_l$:
\begin{equation}
\begin{array}{cl}
\lambda^d_{(l,h)^dj'}(k^2)=\lambda^d_{(h,l)^dj'}(k^2) &\quad \forall j'\in\hat{n}\,, \\
\lambda^d_{(l,h)^d(l',h')^0}(k^2)=\lambda^d_{(h,l)^d(l',h')^0}(k^2) &\quad \forall l'\in\hat{n}, h'\in N_{l'}\,, \\
\lambda^d_{(l,h)^d(l',h')^d}(k^2)=\lambda^d_{(h,l)^d(l',h')^d}(k^2) &\quad \forall l'\in\hat{n}, h'\in N_{l'}\,.
\end{array}
\end{equation}
In other words, all the coefficients $\lambda^d_{(l,h)^d J}(k^2)$ with
$J\in\hat{\mathcal{I}}(k^2)$ are symmetric with respect to an
interchange of $l$ and $h$.

\medskip

\noindent \textit{Step 2.} The sum of derivatives in points
$W_{\{j,k\}}$ is
\begin{equation}\label{Krok2}
g'_{(l,h)}(0)+g'_{(h,l)}(0)=w_{\{l,h\}}\cdot g_{(l,h)}(0)
\end{equation}
for all $l\in\hat{n}, h\in N_l$. We substitute
$$
\left.\frac{\partial\tilde{\GG}^{A}_{\i\kappa}(x,\tilde{y})}{\partial x}\right|_{x=0} =
\frac{\sinh\kappa(d-\tilde{y})}{\sinh\kappa d}\,\e^{-\i
A\tilde{y}}\,,
$$
$$
\left.\left(\e^{\i A x}\sinh\kappa x\right)'\right|_{x=0}=\kappa
$$
and
$$
\left.\left(\e^{\i A x}\sinh\kappa(d-x)\right)'\right|_{x=0}=\kappa\cosh\kappa d-\i A\sinh\kappa d
$$
into Eq.~\eqref{Krok2} and require the equality to be satisfied
for any $\Psi=(\psi_J)_{J\in\mathcal{I}}$. In the course of the
calculation, the outcome of the Step 1 is also used. As a result,
we find how the coefficients $\lambda^d_{(l,h)^d J}(k^2)$
($J\in\mathcal{I}$) can be expressed in terms of
$\lambda^d_{(l,h)^0 J}(k^2)$ and $\lambda^d_{(h,l)^0 J}(k^2)$:
\begin{subequations}
\begin{equation}
\lambda^d_{(l,h)^dj'}(k^2)=\frac{\kappa}{2\kappa\cosh\kappa d+w_{\{l,h\}}\sinh\kappa d}\left(\lambda^d_{(l,h)^0j'}(k^2)+\lambda^d_{(h,l)^0j'}(k^2)\right) \quad \forall j'\in\hat{n}\,,
\end{equation}
\begin{equation}
\lambda^d_{(l,h)^d(l',h')^0}(k^2)=\frac{\kappa}{2\kappa\cosh\kappa d+w_{\{l,h\}}\sinh\kappa d}\left(\lambda^d_{(l,h)^0(l',h')^0}(k^2)+\lambda^d_{(h,l)^0(l',h')^0}(k^2)\right) \quad
\forall l'\in\hat{n}, h'\in N_{l'}\,,
\end{equation}
\begin{multline}
\lambda^d_{(l,h)^d(l',h')^d}(k^2)
=\frac{\kappa}{2\kappa\cosh\kappa d+w_{\{l,h\}}\sinh\kappa d}\left(\lambda^d_{(l,h)^0(l',h')^d}(k^2)+\lambda^d_{(h,l)^0(l',h')^d}(k^2)+\frac{\delta_{(l,h)(l',h')}}{\kappa\sinh\kappa d}\right) \\
\forall l'\in\hat{n}, h'\in N_{l'}
\end{multline}
\end{subequations}
for the indices $l\in\hat{n}, h\in N_l$.

\medskip

\noindent \textit{Step 3.} The continuity at the points $V_j$
requires
\begin{equation}\label{Krok3}
g_j(0)=g_{(j,h)}(d)
\end{equation}
for all $j\in\hat{n}, h\in N_j$. Since $\GG_{\i\kappa}(0,y)=0$ for
all $x\in\R^+$ and $\tilde{\GG}^{A_{(l,h)}}_{\i\kappa}(d)=0$ for
all $\tilde{y}\in(0,d)$, it holds
\begin{multline}
g_j(0)=\sum_{j'=1}^{n}\lambda^d_{jj'}(k^2)\int_{0}^{+\infty}\e^{-\kappa
y} \psi_{j'}(y)\,\D y\cdot\e^{-\kappa x_j} \\
+\sum_{l'=1}^{n}\sum_{h'\in N_{l'}}\left(\lambda^d_{j(l'h')^0}(k^2)\int_{0}^{d}\e^{-\i A_{(l',h')}\tilde{y}}\sinh\kappa
\tilde{y}\psi_{(l',h')}(\tilde{y})\,\D\tilde{y}\right. \\
+\left.\lambda^d_{j(l'h')^d}(k^2)\int_{0}^{d}\e^{-\i
A_{(l',h')}\tilde{y}}\sinh\kappa
(d-\tilde{y})\psi_{(l',h')}(\tilde{y})\,\D\tilde{y}\right)\,,
\end{multline}
and
\begin{multline}
g_{(j,h)}(d)=
\e^{\i A_{(j,h)}d}\cdot\sinh\kappa d\cdot\left[\sum_{j'=1}^{n}\lambda^d_{(j,h)^0j'}(k^2)\int_{0}^{+\infty}\e^{-\kappa
y}\psi_{j'}(y)\,\D y \right.\\
+\sum_{l'=1}^{n}\sum_{h'\in N_{l'}}\left(\lambda^d_{(j,h)^0(l'h')^0}(k^2)\int_{0}^{d}\e^{-\i A_{(l',h')}\tilde{y}}\sinh\kappa
\tilde{y}\psi_{(l',h')}(\tilde{y})\,\D\tilde{y}\right. \\
+\left.\left.\lambda^d_{(j,h)^0(l'h')^d}(k^2)\int_{0}^{d}\e^{-\i
A_{(l',h')}\tilde{y}}\sinh\kappa
(d-\tilde{y})\psi_{(l',h')}(\tilde{y})\,\D\tilde{y}\right)
\right]\,.
\end{multline}
The relation~\eqref{Krok3} should be satisfied for any choice of
$\Psi=(\psi_J)_{J\in\mathcal{I}}$, hence we obtain the
coefficients $\lambda^d_{(j,h)^0J}(k^2)$ with $J\in\mathcal{I}$
expressed in terms of $\lambda^d_{jJ}(k^2)$ in the following way
\begin{subequations}\label{VyslKrok3}
\begin{align}
\lambda^d_{(j,h)^0j'}(k^2)&=\frac{1}{\sinh\kappa d}
\,\e^{-\i dA_{(j,h)}}\cdot\lambda^d_{jj'}(k^2) &\quad \forall j'\in\hat{n}\,, \\
\lambda^d_{(j,h)^0(l',h')^0}(k^2)&=\frac{1}{\sinh\kappa d}
\,\e^{-\i dA_{(j,h)}}\cdot\lambda^d_{j(l',h')^0}(k^2) &\quad \forall l'\in\hat{n}, h'\in N_{l'}\,, \\
\lambda^d_{(j,h)^0(l',h')^d}(k^2)&=\frac{1}{\sinh\kappa
d}\,\e^{-\i dA_{(j,h)}}\cdot\lambda^d_{j(l',h')^d}(k^2) &\quad
\forall l'\in\hat{n}, h'\in N_{l'}
\end{align}
\end{subequations}
for $j\in\hat{n}, h\in N_j$. We also return to the result of Step
2 -- we substitute there for $\lambda^d_{(j,h)^0J}(k^2)$ the
expressions that we have just obtained arriving thus at
\begin{subequations}\label{VyslKrok32}
\begin{equation}
\lambda^d_{(l,h)^dj'}(k^2)=\frac{\kappa}{2\kappa\cosh\kappa d+w_{\{l,h\}}
\sinh\kappa d}\cdot\frac{1}{\sinh\kappa d}\cdot\left(\e^{-\i dA_{(l,h)}}\lambda^d_{lj'}(k^2)+\e^{-\i dA_{(h,l)}}\lambda^d_{hj'}(k^2)\right) \quad
\forall j'\in\hat{n}\,,
\end{equation}
\begin{multline}
\lambda^d_{(l,h)^d(l',h')^0}(k^2)=\frac{\kappa}{2\kappa\cosh\kappa d+w_{\{l,h\}}\sinh\kappa d}\cdot\frac{1}{\sinh\kappa d}
\cdot\left(\e^{-\i dA_{(l,h)}}\lambda^d_{l(l',h')^0}(k^2)+\e^{-\i dA_{(h,l)}}\lambda^d_{h(l',h')^0}(k^2)\right) \\
\forall l'\in\hat{n}, h'\in N_{l'}\,,
\end{multline}
\begin{multline}
\lambda^d_{(l,h)^d(l',h')^d}(k^2)=
\frac{\kappa}{2\kappa\cosh\kappa d+w_{\{l,h\}}\sinh\kappa d}\cdot\frac{1}{\sinh\kappa d}\cdot \\
\cdot\left(\e^{-\i dA_{(l,h)}}\cdot\lambda^d_{l(l',h')^d}(k^2)+\e^{-\i dA_{(h,l)}}\cdot\lambda^d_{h(l',h')^d}(k^2)+\frac{1}{\kappa}\delta_{(l,h)(l',h')}\right) \\
\forall l'\in\hat{n}, h'\in N_{l'}
\end{multline}
\end{subequations}
for $l\in\hat{n}, h\in N_l$.

\medskip

\noindent \textit{Step 4.} In this step we examine the sum of derivatives at the points
$V_j$, i.e. at the junctions of the halflines and connecting
segments. Since the connecting lines support constant vector
potentials, one has to rewrite the original condition into the
form derived in Corollary \ref{Transformovana delta}. Note that
the variable on the connecting segments is considered in the
ingoing sense, thus the sign of the potentials $A_{(j,h)}$ ($h\in
N_j$) has to be taken with the minus sign. The resulting condition
is
\begin{equation}\label{Krok4}
g_j'(0)-\sum_{h\in N_j}g'_{(j,h)}(d)=\left(v_j-\i\sum_{h\in N_j}A_{(j,h)}\right)\cdot g_j(0)
\end{equation}
for all $j\in\hat{n}$.

The way how to proceed in this step is essentially the same as in
previous steps, only the calculus is slightly longer. With the aid
of the formul\ae
\begin{gather*}
\left.\frac{\partial\GG_{\i\kappa}(x,y)}{\partial x}\right|_{x=0}=\e^{-\kappa y}\,, \\
\left.\frac{\partial\tilde{\GG}^{A}_{\i\kappa}(x,\tilde{y})}{\partial
x}\right|_{x=d}= -\frac{\e^{\i Ad}}{\sinh\kappa
d}\cdot\sinh\kappa\tilde{y} \:\e^{-\i A\tilde{y}}
\end{gather*}
and
\begin{gather*}
\left.\left(\e^{\i A x}\sinh\kappa x\right)'\right|_{x=d}
=\left(\kappa\cosh\kappa d+\i A\sinh\kappa d\right)\cdot\e^{\i Ad}\,, \\
\left.\left(\e^{\i A
x}\sinh\kappa(d-x)\right)'\right|_{x=d}=-\kappa\,\e^{\i Ad}\,,
\end{gather*}
used in Eq.~\eqref{Krok4}, we arrive at an expression containing
$\psi_1\ldots,\psi_n$ that should be satisfied for any choice of
$\Psi=(\psi_J)_{J\in\mathcal{I}}$. This yields the following three
groups of conditions:
\begin{subequations}
\begin{multline}
\delta_{jj'}-\kappa\lambda^d_{jj'}(k^2)-\sum_{h\in N_j}\left(\kappa\cosh\kappa d
+\i A_{(j,h)}\sinh\kappa d\right)\cdot\e^{\i dA_{(j,h)}}\lambda^d_{(j,h)^0j'}(k^2) \\
+\kappa\sum_{h\in N_j}\left(\e^{\i dA_{(j,h)}}\lambda^d_{(j,h)^dj'}(k^2)\right)=
\left(v_j-\i\sum_{h\in N_j}A_{(j,h)}\right)\lambda^d_{jj'}(k^2)\,,
\end{multline}
\begin{multline}
\delta_{jl'}\frac{\e^{\i dA_{(j,h')}}}{\sinh\kappa d}-\kappa\lambda^d_{j(l,h')^0}(k^2)
-\sum_{h\in N_j}\left(\kappa\cosh\kappa d+\i A_{(j,h)}\sinh\kappa d\right)
\cdot\e^{\i dA_{(j,h)}}\lambda^d_{(j,h)^0(l,h')^0}(k^2) \\
+\kappa\sum_{h\in N_j}\left(\e^{\i dA_{(j,h)}}\lambda^d_{(j,h)^d(l,h')^0}(k^2)\right)=
\left(v_j-\i\sum_{h\in N_j}A_{(j,h)}\right)\lambda^d_{j(l,h')^0}(k^2)\,,
\end{multline}
\begin{multline}
-\kappa\lambda^d_{j(l,h')^d}(k^2)-\sum_{h\in N_j}\left(\kappa\cosh\kappa d+\i A_{(j,h)}\sinh\kappa d\right)
\cdot\e^{\i dA_{(j,h)}}\lambda^d_{(j,h)^0(l,h')^d}(k^2) \\
+\kappa\sum_{h\in N_j}\left(\e^{\i dA_{(j,h)}}\lambda^d_{(j,h)^d(l,h')^d}(k^2)\right)=
\left(v_j-\i\sum_{h\in N_j}A_{(j,h)}\right)\lambda^d_{j(l,h')^d}(k^2)\,.
\end{multline}
\end{subequations}
We use the equalities~\eqref{VyslKrok3} and \eqref{VyslKrok32} to
eliminate all terms of the type $\lambda^d_{(j,h)^0J}(k^2)$ and
$\lambda^d_{(j,h)^dJ}(k^2)$, $J\in\hat{\mathcal{I}}$. In this way we
obtain three independent systems of equations for $\lambda^d_{jj'}(k^2)$
($j,j'\in\hat{n}$), $\lambda^d_{j(l',h')^0}(k^2)$
($j,l'\in\hat{n},h'\in N_{l'}$) and $\lambda^d_{j(l',h')^d}(k^2)$
($j,l'\in\hat{n},h'\in N_{l'}$):
\begin{subequations}
\begin{equation}\label{Kr4-}
\delta_{jj'}-\kappa\lambda^d_{jj'}(k^2)-\kappa\#N_j\frac{\cosh\kappa d}{\sinh\kappa d}\lambda^d_{jj'}(k^2)
+\sum_{h\in N_j}\frac{\kappa^2}{\sinh\kappa d}\cdot\frac{\lambda^d_{jj'}(k^2)
+\e^{2\i dA_{(j,h)}}\lambda^d_{hj'}(k^2)}{2\kappa\cosh\kappa d+w_{\{j,h\}}\sinh\kappa d}=v_j\lambda^d_{jj'}(k^2)\,,
\end{equation}
\begin{multline}\label{Kr40}
\delta_{jl'}\frac{\e^{\i dA_{(l',h')}}}{\sinh\kappa d}-\kappa\lambda^d_{j(l,h')^0}(k^2)
-\kappa\#N_j\frac{\cosh\kappa d}{\sinh\kappa d}\lambda^d_{j(l,h')^0}(k^2) \\
+\sum_{h\in N_j}\frac{\kappa^2}{\sinh\kappa d}\cdot\frac{\lambda^d_{j(l,h')^0}(k^2)
+\e^{2\i dA_{(j,h)}}\lambda^d_{h(l,h')^0}(k^2)}{2\kappa\cosh\kappa d+w_{\{j,h\}}\sinh\kappa d}=v_j\lambda^d_{j(l,h')^0}(k^2)\,,
\end{multline}
\begin{multline}\label{Kr4d}
-\kappa\lambda^d_{j(l,h')^d}(k^2)-\kappa\#N_j\frac{\cosh\kappa d}{\sinh\kappa d}\lambda^d_{j(l,h')^d}(k^2)
+\sum_{h\in N_j}\frac{\kappa^2}{\sinh\kappa d}\cdot\frac{\lambda^d_{j(l,h')^d}(k^2)
+\e^{2\i dA_{(j,h)}}\lambda^d_{h(l,h')^d}(k^2)}{2\kappa\cosh\kappa d+w_{\{j,h\}}\sinh\kappa d} \\
+\frac{\kappa}{\sinh\kappa d}\cdot\sum_{h\in N_j}\frac{\e^{\i dA_{(j,h)}}}{2\kappa\cosh\kappa d
+w_{\{j,h\}}\sinh\kappa d}\delta_{jl'}\delta_{hh'}=v_j\lambda^d_{j(l,h')^d}(k^2)\,.
\end{multline}
\end{subequations}
Let us focus, e.g., on Eq.~\eqref{Kr4-}, which can be rewritten in
the form
\begin{multline}\label{Kr4--}
\sum_{h=1}^n\left[\delta_{jh}\left(\kappa+\kappa\#N_j\frac{\cosh\kappa d}{\sinh\kappa d}
-\frac{\kappa}{\sinh\kappa d}\sum_{\tilde{h}\in N_j}
\frac{\kappa}{2\kappa\cosh\kappa d+w_{\{j,\tilde{h}\}}\sinh\kappa d}+v_j\right)\right.\\
\left.-\chi_{N_j}(h)\cdot\frac{\kappa}{\sinh\kappa
d}\cdot\frac{\e^{2\i dA_{(j,h)}}}{2\kappa\cosh\kappa d
+w_{\{j,h\}}\sinh\kappa d}\right]\lambda^d_{hj'}(k^2)=\delta_{jj'}
\end{multline}
for all $j,j'\in\hat{n}$; the symbol $\chi_{N_j}(h)$ is equal to
one if $h\in N_j$ holds and zero otherwise. As we will see within
a short time, it is convenient to introduce a matrix $M_d$ the
$(j,h)$-th element of which is defined by
\begin{multline}\label{M}
[M_d]_{jh}=\delta_{jh}\left(\kappa+\kappa\#N_j\frac{\cosh\kappa d}{\sinh\kappa d}
-\frac{\kappa}{\sinh\kappa d}\sum_{\tilde{h}\in N_j}\frac{\kappa}{2\kappa\cosh\kappa d
+w_{\{j,\tilde{h}\}}\sinh\kappa d}+v_j\right)\\
-\chi_{N_j}(h)\cdot\frac{\kappa}{\sinh\kappa d}\cdot\frac{\e^{2\i dA_{(j,h)}}}{2\kappa\cosh\kappa d+w_{\{j,h\}}\sinh\kappa d}\,.
\end{multline}
We also rewrite the set $\hat{\mathcal{I}}$ as a union,
$$
\hat{\mathcal{I}}=\hat{n}\cup \mathcal{J}^0\cup \mathcal{J}^d\,,
$$
where $\mathcal{J}^0=\left\{\left.(l,h)^0\right|\,l\in\hat{n},h\in
N_l\right\}$ and
$\mathcal{J}^d=\left\{\left.(l,h)^d\right|\,l\in\hat{n},h\in
N_l\right\}$, and define the symbols
$\Lambda^{\mathrm{Ag},d}_{XY}(k^2)$ for
$X,Y\in\left\{\hat{n},\mathcal{J}^0,\mathcal{J}^d\right\}$ by the
relation
$$
\Lambda^{\mathrm{Ag},d}_{XY}(k^2)=\left(\lambda^{d}_{JL}(k^2)\right)_{J\in
X,L\in Y}\,,
$$
e.g.
$\Lambda^{\mathrm{Ag},d}_{\hat{n}\mathcal{J}^0}(k^2)=\left(\lambda^{d}_{j(l',h')^0}(k^2)\right)_{j\in\hat{n},
(l',h')^0\in \mathcal{J}^0}$. Obviously, the matrix
$\Lambda^{\mathrm{Ag},d}(k^2)$ has the block structure
$$
\Lambda^{\mathrm{Ag},d}(k^2)=\left(\begin{array}{c|c|c}
\Lambda^{\mathrm{Ag},d}_{\hat{n}\hat{n}}(k^2) & \Lambda^{\mathrm{Ag},d}_{\hat{n}\mathcal{J}^0}(k^2)
& \Lambda^{\mathrm{Ag},d}_{\hat{n}\mathcal{J}^d}(k^2) \\
\hline
\Lambda^{\mathrm{Ag},d}_{\mathcal{J}^0\hat{n}}(k^2) & \Lambda^{\mathrm{Ag},d}_{\mathcal{J}^0\mathcal{J}^0}(k^2)
& \Lambda^{\mathrm{Ag},d}_{\mathcal{J}^0\mathcal{J}^d}(k^2) \\
\hline \Lambda^{\mathrm{Ag},d}_{\mathcal{J}^d\hat{n}}(k^2) &
\Lambda^{\mathrm{Ag},d}_{\mathcal{J}^d\mathcal{J}^0}(k^2) &
\Lambda^{\mathrm{Ag},d}_{\mathcal{J}^d\mathcal{J}^d}(k^2)
\end{array}\right)\,.
$$
We observe that the system of equations \eqref{Kr4--} is nothing
but
$$
M_d\Lambda^{\mathrm{Ag},d}_{\hat{n}\hat{n}}(k^2)=I\,,
$$
and therefore
$\Lambda^{\mathrm{Ag},d}_{\hat{n}\hat{n}}(k^2)=(M_d)^{-1}$, or in
the components
\begin{subequations}
\begin{equation}\label{lambda_nn}
\lambda^{d}_{jj'}(k^2)=\left[(M_d)^{-1}\right]_{jj'}\,.
\end{equation}
The matrices $\Lambda^{\mathrm{Ag},d}_{\hat{n}\mathcal{J}^0}(k^2)$
and $\Lambda^{\mathrm{Ag},d}_{\hat{n}\mathcal{J}^d}(k^2)$ can be
found in a similar way. We start from Eqs.~\eqref{Kr40} and
\eqref{Kr4d} and arrive at
\begin{equation}\label{lambda_n0}
\lambda^{d}_{j(l',h')^0}(k^2)=\frac{\e^{\i dA_{(l',h')}}}{\sinh\kappa d}\left[(M_d)^{-1}\right]_{jl'}
\end{equation}
and
\begin{equation}\label{lambda_nd}
\lambda^{d}_{j(l',h')^d}(k^2)=\frac{\kappa}{\sinh\kappa d}
\cdot\frac{\e^{\i dA_{(l',h')}}}{2\kappa\cosh\kappa d+w_{\{l',h'\}}
\sinh\kappa d}\left[(M_d)^{-1}\right]_{jl'}\,.
\end{equation}
To obtain expressions for
$\Lambda^{\mathrm{Ag},d}_{\mathcal{J}^0X}(k^2)$ and
$\Lambda^{\mathrm{Ag},d}_{\mathcal{J}^dX}(k^2)$
($X=\mathcal{J}^0,\mathcal{J}^d$) we substitute \eqref{lambda_nn},
\eqref{lambda_n0} and \eqref{lambda_nd} into Equations
\eqref{VyslKrok3} and \eqref{VyslKrok32} which gives
\begin{gather}
\lambda^d_{(l,h)^0j'}(k^2)=\frac{\e^{-\i dA_{(l,h)}}}{\sinh\kappa d}
\cdot\left[(M_d)^{-1}\right]_{lj'}\,, \label{lambda_0n}\\
\lambda^d_{(l,h)^0(l',h')^0}(k^2)=\frac{\e^{-\i dA_{(l,h)}}}{\sinh\kappa d}
\cdot\frac{\e^{\i dA_{(l',h')}}}{\sinh\kappa d}\left[(M_d)^{-1}\right]_{ll'}\,, \label{lambda_00} \\
\lambda^d_{(l,h)^0(l',h')^d}(k^2)=\frac{\e^{-\i dA_{(l,h)}}}{\sinh^2\kappa d}
\cdot\frac{\kappa\cdot\e^{\i dA_{(l',h')}}}{2\kappa\cosh\kappa d+w_{\{l',h'\}}\sinh\kappa d}\left[(M_d)^{-1}\right]_{ll'} \label{lambda_0d}\,,
\end{gather}
\begin{equation}\label{lambda_dn}
\lambda^d_{(l,h)^dj'}(k^2)=\frac{\kappa}{2\kappa\cosh\kappa d+w_{\{l,h\}}\sinh\kappa d}\cdot\frac{1}{\sinh\kappa d}
\cdot\left(\e^{-\i dA_{(l,h)}}\left[(M_d)^{-1}\right]_{lj'}+\e^{-\i dA_{(h,l)}}\left[(M_d)^{-1}\right]_{hj'}\right)\,,
\end{equation}
\begin{multline} \label{lambda_d0}
\lambda^d_{(l,h)^d(l',h')^0}(k^2)=\frac{\kappa}{2\kappa\cosh\kappa d+w_{\{l,h\}}\sinh\kappa d}\cdot\frac{\e^{\i dA_{(l',h')}}}{\sinh^2\kappa d}\cdot \\
\cdot\left(\e^{-\i dA_{(l,h)}}\cdot\left[(M_d)^{-1}\right]_{ll'}+\e^{-\i dA_{(h,l)}}\cdot\left[(M_d)^{-1}\right]_{hl'}\right)\,,
\end{multline}
\begin{multline} \label{lambda_dd}
\lambda^d_{(l,h)^d(l',h')^d}(k^2)=
\frac{\kappa}{2\kappa\cosh\kappa d+w_{\{l,h\}}\sinh\kappa d}\cdot\frac{1}{\sinh\kappa d}\cdot\\
\cdot\left[\frac{\kappa}{\sinh\kappa d}\cdot\frac{\e^{\i dA_{(l',h')}}}{2\kappa\cosh\kappa d+w_{\{l',h'\}}\sinh\kappa d}
\left(\e^{-\i dA_{(l,h)}}\cdot\left[(M_d)^{-1}\right]_{ll'}+\e^{-\i dA_{(h,l)}}
\cdot\left[(M_d)^{-1}\right]_{hl'}\right)+\frac{1}{\kappa}\delta_{(l,h)(l',h')}\right].
\end{multline}
\end{subequations}
Once we compute the elements of $(M_d)^{-1}$ explicitly, we will
have fully explicit formulae for $\Lambda^{\mathrm{Ag},d}(k^2)$ we
need. We start from the matrix $M_d$ itself. We take the
formula~\eqref{M}, substitute there the expressions for $v_j(d)$,
$w_{\{j,k\}}(d)$ and $A_{(j,k)}(d)$ that have been obtained
heuristically in the previous section and apply Taylor expansions
to appropriate orders. A slightly laborious calculation leads to
the formulae written below; note that the structure of the
expression for the \mbox{$(j,h)$-th} element of the matrix $M_d$
depends on whether $j,h$ belong to $\hat{m}$ or to
$\hat{n}\backslash\hat{m}$:
\begin{subequations}
\begin{equation}
[M_d]_{jh}=\delta_{jh}\left(\kappa+S_{jj}+\frac{1}{d}\sum_{l=m+1}^n T_{jl}\overline{T_{lj}}\right)+
S_{jh}+\frac{1}{d}\sum_{l=m+1}^n T_{jl}\overline{T_{hl}}+\O(d)\quad
\text{for}\quad j,h\in\hat{m}\,,
\end{equation}
\begin{equation}
[M_d]_{jh}=-\frac{1}{d}T_{jh}+\O(d) \quad\text{for}\quad j\in\hat{m}, h\geq m+1\,,
\end{equation}
\begin{equation}
[M_d]_{jh}=-\frac{1}{d}\overline{T_{hj}}+\O(d) \quad\text{for}\quad j\geq m+1, h\in\hat{m}\,,
\end{equation}
\begin{equation}
[M_d]_{jh}=\delta_{jh}\left(\kappa+\frac{1}{d}\right)+\O(d) \quad\text{for}\quad j,h\geq m+1\,.
\end{equation}
\end{subequations}
The matrix $M_d$ has thus the form
\begin{equation*}
M_d=\left(\begin{array}{c|c} S+\kappa I^{(m)}+\frac{1}{d}TT^* &
-\frac{1}{d}T \\ [.2em] \hline \\ [-.9em] -\frac{1}{d}T^* &
\left(\kappa+\frac{1}{d}\right) I^{(n-m)}
\end{array}\right)+\O(d)\,,
\end{equation*}
where $\O(d)$ on the eight-hand side represents a matrix $n\times
n$ the all entries of which are of order of $\O(d)$ as $d\to 0$.

Our aim is to find the inverse of $M_d$. For this purpose, we
denote the first term on the right-hand side, the principal one,
as $M_{d,P}$ and notice that if $M_{d,P}$ is regular, then
\begin{multline}\label{InverzePerturbovaneMatice}
[M_d]^{-1}=\left(M_{d,P}+\O(d)\right)^{-1}=\left[M_{d,P}(I+[M_{d,P}]^{-1}\O(d))\right]^{-1}=\\
=\left[I-[M_{d,P}]^{-1}\O(d)\right][M_{d,P}]^{-1}=[M_{d,P}]^{-1}-[M_{d,P}]^{-1}\O(d)[M_{d,P}]^{-1}\,.
\end{multline}
Moreover, if $[M_{d,P}]^{-1}=\O(1)$ it obviously holds
$[M_d]^{-1}=[M_{d,P}]^{-1}+\O(d)$; in other words, under certain
assumptions it suffices to find the inverse of $M_{d,P}$.

Since the matrix $M_{d,P}$ has a block structure, one can find
$M_{d,P}$ in the same block structure. This in other words means
that we are looking for a matrix $\left(\begin{array}{c|c}
N_1 & N_2 \\
\hline
N_3 & N_4
\end{array}\right)$
such that the relation
$$
\left(\begin{array}{c|c}
N_1 & N_2 \\
\hline
N_3 & N_4
\end{array}\right)\cdot
\left(\begin{array}{c|c}
S+\kappa I^{(m)}+\frac{1}{d}TT^* & -\frac{1}{d}T \\
\hline
-\frac{1}{d}T^* & \left(\kappa+\frac{1}{d}\right) I^{(n-m)}
\end{array}\right)=
\left(\begin{array}{c|c}
I^{(m)} & 0 \\
\hline
0 & I^{(n-m)}
\end{array}\right)
$$
holds true. It turns out that
\begin{multline*}
[M_{d,P}]^{-1}=\\
\left(\begin{array}{cc}
(S+\kappa I^{(m)}+\frac{\kappa}{1+\kappa d} TT^*)^{-1} &
\frac{1}{1+\kappa d}(S+\kappa I^{(m)}+\frac{\kappa}{1+\kappa d} TT^*)^{-1}T \\
\frac{1}{1+\kappa d} T^*(S+\kappa I^{(m)}+\frac{\kappa}{1+\kappa d} TT^*)^{-1} &
\frac{1}{(1+\kappa d)^2} T^*(S+\kappa I^{(m)}
+\frac{\kappa}{1+\kappa d} TT^*)^{-1}T+\frac{d}{1+\kappa d} I^{(n-m)}
\end{array}\right)
\end{multline*}
provided the matrix $S+\kappa I^{(m)}+\frac{\kappa}{1+\kappa d}
TT^*$ is regular.  Since $\frac{\kappa}{1+\kappa d}=\kappa+\O(d)$,
we may proceed in the same way as
in~\eqref{InverzePerturbovaneMatice}, and we obtain
\begin{equation*}
\left(S+\kappa I^{(m)}+\frac{\kappa}{1+\kappa d} TT^*\right)^{-1}
=(S+\kappa I^{(m)}+\kappa TT^*+\O(d))^{-1}
=(S+\kappa I^{(m)}+\kappa TT^*)^{-1}+\O(d)\,,
\end{equation*}
if the matrix $S+\kappa I^{(m)}+\kappa TT^*$ is regular. However,
the regularity of this matrix has been discussed and proven for an
appropriate $\kappa$ at the end of the part devoted to the
approximated system, see Eq.~\eqref{RegularitaMatice} and the
paragraph following it. It follows that $M_{d,P}$ is regular as
well, i.e. the condition on the regularity of $M_{d,P}$ in
\eqref{InverzePerturbovaneMatice} is satisfied. Hence
$$
[M_{d,P}]^{-1}=
\left(\begin{array}{cc}
(S+\kappa I^{(m)}+\kappa TT^*)^{-1} & (S+\kappa I^{(m)}+\kappa TT^*)^{-1}T \\
T^*(S+\kappa I^{(m)}+\kappa TT^*)^{-1} & T^*(S+\kappa I^{(m)}+\kappa TT^*)^{-1}T
\end{array}\right)+\O(d)\,,
$$
and together with \eqref{InverzePerturbovaneMatice} we have
\begin{equation*}
[M_d]^{-1}=
\left(\begin{array}{cc}
(S+\kappa I^{(m)}+\kappa TT^*)^{-1} & (S+\kappa I^{(m)}+\kappa TT^*)^{-1}T \\
T^*(S+\kappa I^{(m)}+\kappa TT^*)^{-1} & T^*(S+\kappa I^{(m)}+\kappa TT^*)^{-1}T
\end{array}\right)+\O(d)\,.
\end{equation*}
It is important to notice that
\begin{equation}\label{Rad invM}
[M_d]^{-1}=\O(1) \quad\text{for}\; d\to 0_+\,.
\end{equation}
Combining the above result with Eq.~\eqref{lambda_nn}, we can
conclude that
\begin{equation*}
\Lambda^{\mathrm{Ag},d}_{\hat{n}\hat{n}}(k^2)=
\left(\begin{array}{cc}
(S+\kappa I^{(m)}+\kappa TT^*)^{-1} & (S+\kappa I^{(m)}+\kappa TT^*)^{-1}T \\
T^*(S+\kappa I^{(m)}+\kappa TT^*)^{-1} & T^*(S+\kappa I^{(m)}+\kappa TT^*)^{-1}T
\end{array}\right)+\O(d)\,,
\end{equation*}
hence
\begin{equation}\label{rozdil lambd n,n}
\Lambda^{\mathrm{Ag},d}_{\hat{n}\hat{n}}(k^2)=\Lambda^\mathrm{Ad}(k^2)+\O(d)\,.
\end{equation}

Having the coefficient matrix we can determine the resolvent
kernel. First we introduce symbol
$\mathcal{J}=\left\{\left.(l,h)\right|\,l\in\hat{n},h\in
N_l\right\}$ (i.e. $\mathcal{I}=\hat{n}\cup\mathcal{J}$), then we
employ a notation similar to the case of the matrix
$\Lambda^{\mathrm{Ag},d}(k^2)$ and its submatrices. We introduce
symbols $\GG^{\mathrm{Ag},d}_{XY,k}$ for any pair
$X,Y\in\left\{\hat{n},\mathcal{J}\right\}$ to denote the blocks
$\GG^{\mathrm{Ag},d}_{k,XY}=\left(\GG^{\mathrm{Ag},d}_{k,JL}\right)_{J\in
X,L\in Y}$; then the integral kernel $\GG^{\mathrm{Ag},d}_k$ of
$R^\mathrm{Ag}_d(k^2)$ has the structure
\begin{equation}\label{G Ag,d}
\GG^{\mathrm{Ag},d}_k(x,y)=\left(\begin{array}{c|c}
\GG^{\mathrm{Ag},d}_{k,\hat{n}\hat{n}}(x,y) & \GG^{\mathrm{Ag},d}_{k,\hat{n}\mathcal{J}}(x,y) \\
\hline \GG^{\mathrm{Ag},d}_{k,\mathcal{J}\hat{n}}(x,y) &
\GG^{\mathrm{Ag},d}_{k,\mathcal{J}\mathcal{J}}(x,y)
\end{array}\right)
\end{equation}
for $x,y\in G_d$. Using \eqref{G Ad,d} we can write the difference
in question as
\begin{equation}\label{rozdil G}
\GG^{\mathrm{Ag},d}_{\i\kappa}-\GG^{Ad,d}_{\i\kappa}=
\left(\begin{array}{c|c}
\GG^{\mathrm{Ag},d}_{\i\kappa,\hat{n}\hat{n}}
-\GG^\mathrm{Ad}_{\i\kappa} & \GG^{\mathrm{Ag},d}_{\i\kappa,\hat{n}\mathcal{J}} \\
\hline \GG^{\mathrm{Ag},d}_{\i\kappa,\hat{n}\mathcal{J}} &
\GG^{\mathrm{Ag},d}_{\i\kappa,\mathcal{J}\mathcal{J}}
\end{array}\right)\,.
\end{equation}

\medskip

\noindent \textit{III. Comparison of the resolvents}

\medskip

\noindent To make use of the above results we compute first
explicit expressions for all the entries of
$\GG^{\mathrm{Ag},d}_k(x,y)$, up to the error term in the lambda
coefficients indicated in \eqref{rozdil lambd n,n}. They may be
derived from Eqs~\eqref{g aproximujici} together with~\eqref{phi}:
\begin{subequations}
\begin{equation}\label{G Ag,d n n}
\GG^{\mathrm{Ag},d}_{\i\kappa,jj'}(x_j,y_{j'})=
\delta_{jj'}\frac{\sinh\kappa x_<\: \e^{-\kappa x_>}}{\kappa}
+\lambda^d_{jj'}(k^2)\ \e^{-\kappa x_j}\e^{-\kappa y_{j'}}
\quad\text{for}\quad j,j'\in\hat{n}\,,
\end{equation}
\begin{multline}\label{G Ag,d n J}
\GG^{\mathrm{Ag},d}_{\i\kappa,j(l',h')}(x_j,y_{(l',h')})
=\e^{-\kappa x_j}\cdot\e^{-\i A_{(l',h')}y_{(l',h')}}\cdot\left[\lambda^d_{j(l',h')^0}(k^2)\sinh\kappa y_{(l',h')}
+\lambda^d_{j(l',h')^d}(k^2)\sinh\kappa(d-y_{(l',h')})\right]\\
\text{for}\quad j\in\hat{n}\,,(l',h')\in\mathcal{J}\,,
\end{multline}
\begin{multline}\label{G Ag,d J n}
\GG^{\mathrm{Ag},d}_{\i\kappa,(l,h)j'}(x_{(l,h)},y_{j'})
=\e^{\i A_{(l,h)}x_{(l,h)}}\cdot\left[\lambda^d_{(l,h)^0j'}(k^2)\sinh\kappa x_{(l,h)}
+\lambda^d_{(l,h)^dj'}(k^2)\sinh\kappa(d-x_{(l,h)})\right]\cdot\e^{-\kappa y_{j'}}\\
\text{for}\quad (l,h)\in\mathcal{J}\,,j'\in\hat{n}\,,
\end{multline}
\begin{multline}\label{G Ag,d J J}
\GG^{\mathrm{Ag},d}_{\i\kappa,(l,h)(l',h')}(x_{(l,h)},y_{(l',h')})
=\delta_{(l,h)(l',h')}\e^{\i A_{(l,h)}x_{(l,h)}}\frac{\sinh\kappa x_< \sinh\kappa(d-x_>)}
{\kappa\sinh\kappa d}\ \e^{-\i A_{(l',h')}y_{(l',h')}}\\
+\e^{\i A_{(l,h)}x_{(l,h)}}\cdot\sinh\kappa x_{(l,h)}\cdot\e^{-\i A_{(l',h')}y_{(l',h')}}
\cdot\left[\lambda^d_{(l,h)^0(l'h')^0}\sinh\kappa y_{(l',h')}
+\lambda^d_{(l,h)^0(l'h')^d}\sinh\kappa(d-y_{(l',h')})\right]\\
+\e^{\i A_{(l,h)}x_{(l,h)}}\cdot\sinh\kappa(d-x_{(l,h)})\cdot\e^{-\i A_{(l',h')}y_{(l',h')}}
\cdot\left[\lambda^d_{(l,h)^d(l'h')^0}\sinh\kappa
y_{(l',h')}+\lambda^d_{(l,h)^d(l'h')^d}\sinh\kappa
(d-y_{(l',h')})\right] \\
\text{for}\quad (l,h),(l',h')\in\mathcal{J}\,.
\end{multline}
\end{subequations}

Now we are able to compare the entries of
$\GG^{\mathrm{Ag},d}_k(x,y)$ given by~\eqref{G Ag,d} and
$\GG^{Ad,d}_k(x,y)$ as specified in~\eqref{G Ad,d}). We begin with
the upper left submatrix $n\times n$ of~\eqref{rozdil G}. From the
expressions for $\GG^\mathrm{Ad}_{\i\kappa}$ and
$\GG^{\mathrm{Ag},d}_{\i\kappa,\hat{n}\hat{n}}$, cf. \eqref{G Ad}
and \eqref{G Ag,d n n}, we have
\begin{multline}
\left[\GG^{\mathrm{Ag},d}_{\i\kappa,\hat{n}\hat{n}}-\GG^\mathrm{Ad}_{\i\kappa}\right]_{jj'}(x_j,y_{j'})\\
=\delta_{jj'}\frac{\sinh\kappa x_<\: \e^{-\kappa x_>}}{\kappa}
+\lambda^d_{jj'}(k^2)\ \e^{-\kappa x_j}\e^{-\kappa y_{j'}}
-\left[\delta_{jl}\frac{\sinh\kappa x_<\: \e^{-\kappa x_>}}{\kappa}
+\lambda_{jj'}(k^2)\ \e^{-\kappa x_j}\e^{-\kappa y_{j'}}\right]\\
=\left(\lambda^d_{jj'}(k^2)-\lambda_{jj'}(k^2)\right)\ \e^{-\kappa x_j}\e^{-\kappa y_{j'}}
=\left[\Lambda^{\mathrm{Ag},d}_{\hat{n}\hat{n}}(k^2)-\Lambda^\mathrm{Ad}(k^2)\right]_{jj'}\ \e^{-\kappa x_j}\e^{-\kappa y_{j'}}
=\O(d)\ \e^{-\kappa x_j}\e^{-\kappa y_{j'}}\,,
\end{multline}
where the last equality holds by virtue of \eqref{rozdil lambd
n,n}. Since such estimate is valid for all $j,j'\in\hat{n}$, there
is a constant $K_1$ independent of $j,j'$ and $d$ such that
\begin{equation}\label{K1}
\left|\GG^{\mathrm{Ag},d}_{\i\kappa,jj'}(x_j,y_{j'})-\GG^\mathrm{Ad}_{\i\kappa,jj'}(x_j,y_{j'})\right|
<K_1\,d\:\e^{-\kappa x_j}\e^{-\kappa y_{j'}}
\end{equation}
holds for all $j,j'\in\hat{n}$, $x_j,y_{j'}\in\R^+$ and any $d$
sufficiently small.

Then we proceed to the upper right submatrix of~\eqref{rozdil G}.
To find a bound for the entries of
$\GG^{\mathrm{Ag},d}_{k,\hat{n}\mathcal{J}}(x,y)$, cf.~\eqref{G
Ag,d n J}, we substitute values of $\lambda^d_{j(l',h')^0}(k^2)$ and
$\lambda^d_{j(l',h')^d}(k^2)$ that we have obtained in
\eqref{lambda_n0} and \eqref{lambda_nd}:
\begin{multline*}
\GG^{\mathrm{Ag},d}_{\i\kappa,j(l',h')}(x_j,y_{(l',h')})=\\
\e^{-\kappa x_j}\cdot\e^{-\i A_{(l',h')}y_{(l',h')}}\cdot\e^{\i dA_{(l',h')}}
\cdot\left[\frac{\sinh\kappa y_{(l',h')}}{\sinh\kappa d}+
\frac{\kappa}{\sinh\kappa d}\cdot\frac{\sinh\kappa(d-y_{(l',h')})}{2\kappa\cosh\kappa d
+w_{\{l',h'\}}\sinh\kappa d}\right]\cdot\left[(M_d)^{-1}\right]_{jl'}\,.
\end{multline*}
It holds $\left[(M_d)^{-1}\right]_{jl'}=\O(1)$ by virtue
of~\eqref{Rad invM} and obviously $\left|\e^{-\i
A_{(l',h')}y_{(l',h')}}\right|=\left|\e^{\i
dA_{(l',h')}}\right|=1$, thus it suffices to estimate the terms in
the brackets. When $d$ is sufficiently small, it holds
$\left|\frac{\sinh\kappa y_{(l',h')}}{\sinh\kappa d}\right|<1$,
because $0<y_{(l',h')}<d$; similarly
$\left|\frac{\sinh\kappa(d-y_{(l',h')})}{\sinh\kappa d}\right|<1$.
As for the denominator of the second term, we substitute for
$w_{\{l',h'\}}$ from \eqref{w m,m} or \eqref{w m,n-m}, depending
on whether both $l',h'$ belong to $\hat{m}$ or not, and we easily
obtain the estimate
$$
\frac{1}{2\kappa\cosh\kappa d+w_{\{l',h'\}}\sinh\kappa d}=\O(1)\,.
$$
Summing all this up, we get
$$
\GG^{\mathrm{Ag},d}_{\i\kappa,j(l',h')}(x_j,y_{(l',h')})=\e^{-\kappa
x_j}\left(\O(1)+\O(1)\right)=\e^{-\kappa x_j}\O(1)
$$
independently of $j$, $(l',h')$ and $x,y$, thus there is a
constant $K_2$ independent of $d$ such that
\begin{equation}\label{K2}
\left|\GG^{\mathrm{Ag},d}_{\i\kappa,j(l',h')}(x_j,y_{(l',h')})\right|<K_2\,\e^{-\kappa
x_j}
\end{equation}
for all $j\in\hat{n}$, $(l',h')\in\mathcal{J}$, $x_j\in\R^+$,
$\,y_{(l',h')}\in(0,d)$ and $d$ sufficiently small.

Similarly we proceed in the case of the left and right bottom
submatrices of~\eqref{rozdil G}, i.e. when estimating the entries
of $\GG^{\mathrm{Ag},d}_{k,\mathcal{J}\hat{n}}(x,y)$ and
$\GG^{\mathrm{Ag},d}_{k,\mathcal{J}\mathcal{J}}(x,y)$. As for
$\GG^{\mathrm{Ag},d}_{k,\mathcal{J}\hat{n}}(x,y)$, we substitute
for $\lambda^d_{(l,h)^0j'}(k^2)$ and $\lambda^d_{(l,h)^dj'}(k^2)$ from
\eqref{lambda_0n} and \eqref{lambda_dn} into \eqref{G Ag,d J n}
and obtain
\begin{multline*}
\GG^{\mathrm{Ag},d}_{\i\kappa,(l,h)j'}(x_{(l,h)},y_{j'})=
\e^{-\kappa y_{j'}}\cdot\e^{\i A_{(l,h)}x_{(l,h)}}
\cdot\left[\frac{\e^{-\i dA_{(l,h)}}}{\sinh\kappa d}\cdot\left[(M_d)^{-1}\right]_{lj'}\sinh\kappa x_{(l,h)} \right.\\
+\frac{\kappa}{2\kappa\cosh\kappa d+w_{\{l,h\}}\sinh\kappa d}\cdot\frac{1}{\sinh\kappa d}\cdot\sinh\kappa(d-x_{(l,h)})\cdot\\
\left.\cdot\left(\e^{-\i dA_{(l,h)}}\left[(M_d)^{-1}\right]_{lj'}+\e^{-\i dA_{(h,l)}}\left[(M_d)^{-1}\right]_{hj'}\right)\right]\,.
\end{multline*}
Using analogous estimates as in the case of
$\GG^{\mathrm{Ag},d}_{\i\kappa,j(l',h')}(x,y)$ above, we obtain
$$
\GG^{\mathrm{Ag},d}_{\i\kappa,(l,h)j'}(x_{(l,h)},y_{j'})=\e^{-\kappa
y_{j'}}\left(\O(1)+\O(1)\right)=\e^{-\kappa y_{j'}}\O(1)\,,
$$
thus there is a constant $K_3$ independent of $d$ such that
\begin{equation}\label{K3}
\left|\GG^{\mathrm{Ag},d}_{\i\kappa,(l,h)j'}(x_{(l,h)},y_{j'})\right|<K_3\,\e^{-\kappa
y_{j'}}
\end{equation}
for all $(l,h)\in\mathcal{J}$, $j'\in\hat{n}$,
$x_{(l,h)}\in(0,d)$, $y\in\R^+$ and $d$ sufficiently small.

Finally, we substitute from \eqref{lambda_00}, \eqref{lambda_0d},
\eqref{lambda_d0} and \eqref{lambda_dd} for
$\lambda^d_{(l,h)^0(l'h')^0}(k^2)$, $\lambda^d_{(l,h)^0(l'h')^d}(k^2)$,
$\lambda^d_{(l,h)^d(l'h')^0}(k^2)$ and $\lambda^d_{(l,h)^d(l'h')^d}(k^2)$,
respectively, into Eq.~\eqref{G Ag,d J J} and obtain
\begin{multline}
\GG^{\mathrm{Ag},d}_{\i\kappa,(l,h)(l',h')}(x_{(l,h)},y_{(l',h')})
=\delta_{(l,h)(l',h')}\e^{\i A_{(l,h)}x_{(l,h)}}\frac{\sinh\kappa x_< \sinh\kappa(d-x_>)}
{\kappa\sinh\kappa d}\ \e^{-\i A_{(l',h')}y_{(l',h')}}\\
+\e^{\i A_{(l,h)}x_{(l,h)}}\cdot\e^{-\i A_{(l',h')}y_{(l',h')}}
\cdot\e^{-\i dA_{(l,h)}}\cdot\e^{\i dA_{(l',h')}}\cdot\frac{\sinh\kappa x_{(l,h)}}{\sinh\kappa d}\cdot\\
\cdot\left[\frac{\sinh\kappa y_{(l',h')}}{\sinh\kappa d}
+\frac{\kappa}{\sinh\kappa d}\cdot\frac{\sinh\kappa(d-y_{(l',h')})}{2\kappa\cosh\kappa d+w_{\{l',h'\}}\sinh\kappa d}\right]
\cdot\left[(M_d)^{-1}\right]_{ll'}\\
+\e^{\i A_{(l,h)}x_{(l,h)}}\cdot\sinh\kappa(d-x_{(l,h)})
\cdot\e^{-\i A_{(l',h')}y_{(l',h')}}\cdot\frac{\kappa}{2\kappa\cosh\kappa d+w_{\{l,h\}}\sinh\kappa d}\cdot\frac{1}{\sinh\kappa d}\cdot\\
\cdot\left[\e^{\i dA_{(l',h')}}\cdot\frac{\sinh\kappa y_{(l',h')}}
{\sinh\kappa d}\left(\e^{-\i dA_{(l,h)}}\cdot\left[(M_d)^{-1}\right]_{ll'}
+\e^{-\i dA_{(h,l)}}\cdot\left[(M_d)^{-1}\right]_{hl'}\right) \right.\\
+\e^{\i dA_{(l',h')}}\cdot\frac{\kappa}{\sinh\kappa d}\cdot\frac{\sinh\kappa
(d-y_{(l',h')})}{2\kappa\cosh\kappa d+w_{\{l',h'\}}\sinh\kappa d}\cdot\\
\left.\cdot\left(\e^{-\i dA_{(l,h)}}\cdot\left[(M_d)^{-1}\right]_{ll'}
+\e^{-\i dA_{(h,l)}}\cdot\left[(M_d)^{-1}\right]_{hl'}\right)+\frac{1}{\kappa}\delta_{(l,h)(l',h')}\right]\,.
\end{multline}
It obviously holds
\begin{equation*}
\GG^{\mathrm{Ag},d}_{\i\kappa,(l,h)(l',h')}(x_{(l,h)},y_{(l',h')})=\\
=\O(d)+\O(1)\cdot\left[\O(1)+\O(1)\right]+\O(1)\cdot\left[\O(1)+\O(1)\right]=\O(1)\,,
\end{equation*}
thus there is a constant $K_4$ independent of $d$ such that
\begin{equation}\label{K4}
\left|\GG^{\mathrm{Ag},d}_{\i\kappa,(l,h)(l',h')}(x_{(l,h)},y_{(l',h')})\right|<K_4
\end{equation}
for all $(l,h),(l',h')\in\mathcal{J}$,
$x_{(l,h)},y_{(l',h')}\in(0,d)$ and any $d$ sufficiently small.

With the help of \eqref{K1}, \eqref{K2}, \eqref{K3} and
\eqref{K4}, we may now estimate all the entries of \eqref{rozdil
G}, which will allows us to assess the Hilbert-Schmidt norm of the
resolvent difference for the operators $H^\mathrm{Ad}_d$ and
$H^\mathrm{Ag}_d$. This norm can be written explicitly as follows,
\begin{multline*}
\left\|R^\mathrm{Ag}_d(k^2)-R^\mathrm{Ad}_d(k^2)\right\|_2^2
=\sum_{j,j'=1}^{n}\int^{+\infty}_{0}\int_{0}^{+\infty}\left|\GG^{\mathrm{Ag},d}_{\i\kappa,jj'}(x_j,y_{j'})
-\GG^\mathrm{Ad}_{\i\kappa,jj'}(x_j,y_{j'})\right|^2\,\D x_j\D y_{j'}\\
+\sum_{j=1}^{n}\sum_{(l',h')\in\mathcal{I}}
\int_0^{+\infty}\int_0^d\left|\GG^{\mathrm{Ag},d}_{\i\kappa,j(l',h')}(x_j,y_{(l',h')})\right|^2\,\D x_j\D y_{(l',h')}\\
+\sum_{(l,h)\in\mathcal{I}}
\sum_{j'=1}^{n}\int_0^d\int_0^{+\infty}\left|\GG^{\mathrm{Ag},d}_{\i\kappa,(l,h)j'}(x_{(l,h)},y_{j'})\right|^2\,\D x_{(l,h)}\D y_{j'}\\
+\sum_{(l,h)\in\mathcal{I}}
\sum_{(l',h')\in\mathcal{I}}\int_0^d\int_0^d\left|\GG^{\mathrm{Ag},d}_{\i\kappa,(l,h)(l'h')}(x_{(l,h)},y_{(l,h')})\right|^2\,\D
x_{(l,h)}\D y_{(l',h')}\,.
\end{multline*}
Now we employ the estimates derived above obtaining
\begin{multline*}
\left\|R^\mathrm{Ag}_d(k^2)-R^\mathrm{Ad}_d(k^2)\right\|_2^2\leq\\
\sum_{j,j'=1}^{n}\int^{+\infty}_{0}\int_{0}^{+\infty}\left|K_1\,d\:\e^{-\kappa x_j}\e^{-\kappa y_{j'}}\right|^2\,\D x_j\D y_{j'}
+\sum_{j=1}^{n}\sum_{(l',h')\in\mathcal{I}}\int_0^{+\infty}\int_0^d\left|K_2\,\e^{-\kappa x_j}\right|^2\,\D x_j\D y_{(l',h')}\\
+\sum_{(l,h)\in\mathcal{I}}\sum_{j'=1}^{n}\int_0^d\int_0^{+\infty}\left|K_3\,\e^{-\kappa y_{j'}}\right|^2\,\D x_{(l,h)}\D y_{j'}
+\sum_{(l,h)\in\mathcal{I}}\sum_{(l',h')\in\mathcal{I}}\int_0^d\int_0^d\left|K_4\right|^2\,\D
x_{(l,h)}\D y_{(l',h')}
\end{multline*}
\begin{gather*}
\leq\sum_{j,j'=1}^{n}K_1^2d^2\:\int_0^{+\infty}e^{-2(\Re\kappa)x_j}\,\D x_j\int_0^{+\infty}e^{-2(\Re\kappa)y_{j'}}\,\D y_{j'}\\
+\sum_{j=1}^{n}\sum_{(l',h')\in\mathcal{I}}K_2^2\int_0^{+\infty}e^{-2(\Re\kappa)x_j}\,\D x_j\cdot\int_0^d 1\,\D y_{(l',h')}\\
+\sum_{(l,h)\in\mathcal{I}}\sum_{j'=1}^{n}K_3^2\int_0^d 1\,\D x_{(l,h)}\cdot\int_0^{+\infty}e^{-2(\Re\kappa)y_{j'}}\,\D y_{j'}\\
+\sum_{(l,h)\in\mathcal{I}}\sum_{(l',h')\in\mathcal{I}}K_4^2\int_0^d
1\,\D x_{(l,h)}\cdot\int_0^d 1\,\D y_{(l',h')}
\end{gather*}
\begin{multline*}
=\sum_{j,j'=1}^{n}K_1^2d^2\frac{1}{(2\Re\kappa)^2}+\sum_{j=1}^{n}\sum_{(l',h')\in\mathcal{I}}K_2^2\frac{1}{2\Re\kappa}\cdot d
+\sum_{(l,h)\in\mathcal{I}}\sum_{j'=1}^{n}K_3^2 d\cdot\frac{1}{2\Re\kappa}+\sum_{(l,h)\in\mathcal{I}}\sum_{(l',h')\in\mathcal{I}}K_4^2 d^2\\
=\O(d)\,.
\end{multline*}
Hence
$$
\left\|R^\mathrm{Ag}_d(k^2)-R^\mathrm{Ad}_d(k^2)\right\|_2=\O\left(\sqrt{d}\right)
\qquad\text{for}\ d\to0_+\,,
$$
and consequently, the Hilbert-Schmidt norm of the difference
$R^\mathrm{Ag}_d(k^2)-R^\mathrm{Ad}_d(k^2)$ tends to zero as
$d\to0_+$ with the explicit convergence rate. Since the HS norm
dominates the operator one, it follows immediately
$$
\lim_{d\to 0_+}
\left\|R^\mathrm{Ag}_d(k^2)-R^\mathrm{Ad}_d(k^2)\right\|=0\,,
$$
therefore the resolvent difference tends to zero in $L^2(G_d)$ as
$d\to 0_+$, which we set out to prove.
\end{proof}

\section*{Acknowledgments}

P.E. is grateful for the hospitality extended to him at the Kochi
University of Technology where the idea of the approximation was
formulated. The research was supported by the Czech Ministry of
Education, Youth and Sports within the project LC06002 and by the 
Japanese Ministry of Education, Culture, Sports, Science and Technology under the Grant number 21540402.



\begin{thebibliography}{99}



\bibitem{RS53}
K.~Ruedenberg and C.W.~Scherr, Free-electron network model for
conjugated systems, I.~Theory, \textit{J. Chem. Phys.} \textbf{21}
(1953), 1565--1581.

\bibitem{EKST08}
P.~Exner, J.P.~Keating, P.~Kuchment, T.~Sunada, A.~Teplyaev, eds.:
{\em Analysis on Graphs and Applications}, Proceedings of a Isaac
Newton Institute programme, January~8--June~29, 2007; 670 p.; AMS
``Proceedings of Symposia in Pure Mathematics'' Series, vol.~77,
Providence, R.I., 2008

\bibitem{CS98}
T.~Cheon and T.~Shigehara: Realizing discontinuous wave functions
with renormalized short-range potentials, \emph{Phys. Lett.}
\textbf{A243} (1998), 111--116.

\bibitem{ENZ01}
P.~Exner, H.~Neidhardt and V.A.~Zagrebnov: Potential
approximations to $\delta'$: an inverse Klauder phenomenon with
norm-resolvent convergence, {\em Commun. Math. Phys.} {\bf 224}
(2001), 593--612.

\bibitem{CE04}
T. Cheon and P. Exner, An approximation to delta' couplings on
graphs, {\it J. Phys. A: Math. Gen.} {\bf 37} (2004), L329--335.

\bibitem{ET07}
P.~Exner and O.~Turek: Approximations of singular vertex couplings
in quantum graphs, {\it Rev. Math. Phys.} {\bf 19} (2007),
571--606.

\bibitem{SMMC99}
T.~Shigehara, H.~Mizoguchi, T.~Mishima, T.~Cheon: Realization of a
four parameter family of generalized one-dimensional contact
interactions by three nearby delta potentials with renormalized
strengths, \emph{IEICE Trans. Fund. Elec. Comm. Comp. Sci.}
\textbf{E82-A} (1999), 1708--1713.

\bibitem{EP08}
P.~Exner, O.~Post: {\em Approximation of quantum graph vertex
couplings by scaled Schr\"odinger operators on thin branched
manifolds}, \texttt{arXiv: 0811.3707v1}

\bibitem{KS99}
V.~Kostrykin, R.~Schrader: Kirchhoff's rule for quantum wires,
\emph{J. Phys. A: Math. Gen.} \textbf{32} (1999), 595--630.

\bibitem{Ha00}
M. Harmer: Hermitian symplectic geometry and extension theory,
{\it J. Phys. A: Math. Gen.} \textbf{33} (2000), 9193--9203

\bibitem{KS00}
V.~Kostrykin, R.~Schrader: Kirchhoff's rule for quantum wires. II:
The Inverse Problem with Possible Applications to Quantum
Computers, \textit{Fortschr. Phys.} \textbf{48} (2000), 703-716.

\bibitem{GG91}
V.I. Gorbachuk, M.L. Gorbachuk: {\it Boundary value problems for
operator differential equations}, Kluwer, Dordrecht 1991.

\bibitem{FT00}
T. F\"ul\H{o}p, I. Tsutsui: A free particle on a circle with point interaction, 
{\it Phys. Lett.} {\bf A264} (2000),
366--374.

\bibitem{ES89}
P.~Exner and P.~\v{S}eba: Free quantum motion on a branching
graph, \textit{Rep. Math. Phys.} {\bf 28} (1989), 7--26.

\bibitem{Ku04}
P.~Kuchment: Quantum graphs: I. Some basic structures, {\it Waves
Random Media} \textbf{14} (2004), S107--S128.

\bibitem{AGHH}
S.~Albeverio, F.~Gesztesy, R.~H\o egh-Krohn and H.~Holden,
\textit{ Solvable Models in Quantum Mechanics}, 2nd edition, AMS
Chelsea, 2005.
%
\end{thebibliography}
\end{document}